\documentclass[aps,pra,10pt,notitlepage,twocolumn,nofootinbib,superscriptaddress]{revtex4-2}
\usepackage{mathtools}
\usepackage{amsmath}
\usepackage[shortlabels]{enumitem}
\usepackage[dvipsnames]{xcolor}
\usepackage{framed}
\definecolor{shadecolor}{rgb}{0.85,0.90,0.95}

\usepackage{graphicx,epic,eepic,epsfig,amsmath,latexsym,amssymb,verbatim,color}
 
\usepackage{amsfonts}       % blackboard math symbols
\usepackage{nicefrac}       % compact symbols for 1/2, etc.
\usepackage{mathrsfs}

\usepackage{amsmath}
\usepackage{bbm}

\usepackage{float}
\usepackage{tikz}
\usetikzlibrary{chains}
\usetikzlibrary{fit}
\usetikzlibrary{arrows}
\usetikzlibrary{decorations}

\usepackage{epsfig}
\usetikzlibrary{shapes.symbols,patterns} % for source symbols
\usepackage{pgfplots}

\usepackage[strict]{changepage}
\usepackage{hyperref}
\hypersetup{colorlinks=true,citecolor=blue,linkcolor=blue,filecolor=blue,urlcolor=blue,breaklinks=true}

\usepackage[marginal]{footmisc}
\usepackage{url}
\usepackage{theorem}

\newtheorem{definition}{Definition}
\newtheorem{proposition}{Proposition}
\newtheorem{lemma}[proposition]{Lemma}

\newtheorem{theorem}[proposition]{Theorem}

\def\squareforqed{\hbox{\rlap{$\sqcap$}$\sqcup$}}
\def\qed{\ifmmode\squareforqed\else{\unskip\nobreak\hfil
\penalty50\hskip1em\null\nobreak\hfil\squareforqed
\parfillskip=0pt\finalhyphendemerits=0\endgraf}\fi}
\def\endenv{\ifmmode\;\else{\unskip\nobreak\hfil
\penalty50\hskip1em\null\nobreak\hfil\;
\parfillskip=0pt\finalhyphendemerits=0\endgraf}\fi}
\newenvironment{proof}{\noindent \textbf{{Proof~} }}{\hfill $\blacksquare$}

\newcounter{remark}
\newenvironment{remark}[1][]{\refstepcounter{remark}\par\medskip\noindent%
\textbf{Remark~\theremark #1} }{\medskip}

\newcounter{example}

\mathchardef\ordinarycolon\mathcode`\:
\mathcode`\:=\string"8000
\def\vcentcolon{\mathrel{\mathop\ordinarycolon}}
\begingroup \catcode`\:=\active
  \lowercase{\endgroup
  \let :\vcentcolon
  }

\usepackage{cleveref}
\usepackage{graphicx}
\usepackage{xcolor}

\RequirePackage[framemethod=default]{mdframed}
\newmdenv[skipabove=7pt,
skipbelow=7pt,
backgroundcolor=darkblue!15,
innerleftmargin=5pt,
innerrightmargin=5pt,
innertopmargin=5pt,
leftmargin=0cm,
rightmargin=0cm,
innerbottommargin=5pt,
linewidth=1pt]{tBox}

\newmdenv[skipabove=7pt,
skipbelow=7pt,
backgroundcolor=darkred!15,
innerleftmargin=5pt,
innerrightmargin=5pt,
innertopmargin=5pt,
leftmargin=0cm,
rightmargin=0cm,
innerbottommargin=5pt,
linewidth=1pt]{rBox}

\newmdenv[skipabove=7pt,
skipbelow=7pt,
backgroundcolor=blue2!25,
innerleftmargin=5pt,
innerrightmargin=5pt,
innertopmargin=5pt,
leftmargin=0cm,
rightmargin=0cm,
innerbottommargin=5pt,
linewidth=1pt]{dBox}
\newmdenv[skipabove=7pt,
skipbelow=7pt,
backgroundcolor=darkkblue!15,
innerleftmargin=5pt,
innerrightmargin=5pt,
innertopmargin=5pt,
leftmargin=0cm,
rightmargin=0cm,
innerbottommargin=5pt,
linewidth=1pt]{sBox}
\definecolor{darkblue}{RGB}{0,76,156}
\definecolor{darkkblue}{RGB}{0,0,153}
\definecolor{blue2}{RGB}{102,178,255}
\definecolor{darkred}{RGB}{195,0,0}
\newcommand{\nc}{\newcommand}
\nc{\rnc}{\renewcommand}
\nc{\lbar}[1]{\overline{#1}}
\nc{\bra}[1]{\langle#1|}
\nc{\ket}[1]{|#1\rangle}
\nc{\ketbra}[2]{|#1\rangle\!\langle#2|}
\nc{\braket}[2]{\langle#1|#2\rangle}

\nc{\proj}[1]{| #1\rangle\!\langle #1 |}
\nc{\avg}[1]{\langle#1\rangle}
\nc{\rank}{\operatorname{Rank}}
\nc{\smfrac}[2]{\mbox{$\frac{#1}{#2}$}}
\nc{\tr}{\operatorname{Tr}}
\nc{\ox}{\otimes}
\nc{\dg}{\dagger}
\nc{\dn}{\downarrow}
\nc{\cA}{{\cal A}}
\nc{\cB}{{\cal B}}
\nc{\cC}{{\cal C}}
\nc{\cD}{{\cal D}}
\nc{\cE}{{\cal E}}
\nc{\cF}{{\cal F}}
\nc{\cG}{{\cal G}}
\nc{\cH}{{\cal H}}
\nc{\cI}{{\cal I}}
\nc{\cJ}{{\cal J}}
\nc{\cK}{{\cal K}}
\nc{\cL}{{\cal L}}
\nc{\cM}{{\cal M}}
\nc{\cN}{{\cal N}}
\nc{\cO}{{\cal O}}
\nc{\cP}{{\cal P}}
\nc{\cQ}{{\cal Q}}
\nc{\cR}{{\cal R}}
\nc{\cS}{{\cal S}}
\nc{\cT}{{\cal T}}
\nc{\cU}{{\cal U}}
\nc{\cV}{{\cal V}}
\nc{\cX}{{\cal X}}
\nc{\cY}{{\cal Y}}
\nc{\cZ}{{\cal Z}}
\nc{\cW}{{\cal W}}
\nc{\csupp}{{\operatorname{csupp}}}
\nc{\qsupp}{{\operatorname{qsupp}}}
\nc{\var}{{\operatorname{var}}}
\nc{\rar}{\rightarrow}
\nc{\lrar}{\longrightarrow}
\nc{\polylog}{{\operatorname{polylog}}}
\nc{\wt}{{\operatorname{wt}}}
\nc{\av}[1]{{\left\langle {#1} \right\rangle}}
\nc{\supp}{{\operatorname{supp}}}

\nc{\argmin}{{\operatorname{argmin}}}

\def\x{\xi}

\nc{\RR}{{{\mathbb R}}}
\nc{\CC}{{{\mathbb C}}}
\nc{\FF}{{{\mathbb F}}}
\nc{\NN}{{{\mathbb N}}}
\nc{\ZZ}{{{\mathbb Z}}}
\nc{\PP}{{{\mathbb P}}}
\nc{\QQ}{{{\mathbb Q}}}
\nc{\UU}{{{\mathbb U}}}
\nc{\EE}{{{\mathbb E}}}
\nc{\id}{{\operatorname{id}}}

\nc{\CHSH}{{\operatorname{CHSH}}}

\newcommand{\Op}{\operatorname}
\nc{\be}{\begin{equation}}
\nc{\ee}{{\end{equation}}}
\nc{\bea}{\begin{eqnarray}}
\nc{\eea}{\end{eqnarray}}
\nc{\<}{\langle}
\rnc{\>}{\rangle}
\nc{\rU}{\mbox{U}}

\nc{\ob}[1]{#1}

\nc{\OLOCC}{{\text{1-LOCC}}}
\nc{\SEP}{{\text{SEP}}}
\nc{\NS}{{\text{NS}}}
\nc{\LOCC}{{\text{LOCC}}}
\nc{\PPT}{{\text{PPT}}}
\nc{\EXT}{{\text{EXT}}}
\nc{\Sym}{{\operatorname{Sym}}}

\nc{\ERLO}{{E_{{ R,LO}}}}
\nc{\ERLOCC}{{E_{{R,\text{PPT}}}}}
\nc{\ERPPT}{{E_{{R,\text{PPT}}}}}
\nc{\ERPPTinf}{{E^{\infty}_{{R,\text{PPT}}}}}
\nc{\ER}{E_{\rm R}}
\nc{\ERLOCCinfty}{{E^{\infty}_{{r,LOCC}}}}
\nc{\Aram}{{\operatorname{\sf A}}}
\nc{\ECPPT}{{E_{{C,\text{PPT}}}}}
\nc{\EDPPT}{{E_{{D,\text{PPT}}}}}
\nc{\Freek}{{\text{PPT$_k$}}}
\nc{\Freesec}{{\text{PPT$_2$}}}

\nc{\NB}{N}
\nc{\LB}{LN}
\nc{\NPT}{{\text{NPT}}}

\usepackage{tikz}
\usepackage{hyperref}
\hypersetup{colorlinks=true,citecolor=blue,linkcolor=blue,filecolor=blue,urlcolor=blue,breaklinks=true}

\makeatletter
\def\grd@save@target#1{%
  \def\grd@target{#1}}
\def\grd@save@start#1{%
  \def\grd@start{#1}}
\tikzset{
  grid with coordinates/.style={
    to path={%
      \pgfextra{%
        \edef\grd@@target{(\tikztotarget)}%
        \tikz@scan@one@point\grd@save@target\grd@@target\relax
        \edef\grd@@start{(\tikztostart)}%
        \tikz@scan@one@point\grd@save@start\grd@@start\relax
        \draw[minor help lines,magenta] (\tikztostart) grid (\tikztotarget);
        \draw[major help lines] (\tikztostart) grid (\tikztotarget);
        \grd@start
        \pgfmathsetmacro{\grd@xa}{\the\pgf@x/1cm}
        \pgfmathsetmacro{\grd@ya}{\the\pgf@y/1cm}
        \grd@target
        \pgfmathsetmacro{\grd@xb}{\the\pgf@x/1cm}
        \pgfmathsetmacro{\grd@yb}{\the\pgf@y/1cm}
        \pgfmathsetmacro{\grd@xc}{\grd@xa + \pgfkeysvalueof{/tikz/grid with coordinates/major step}}
        \pgfmathsetmacro{\grd@yc}{\grd@ya + \pgfkeysvalueof{/tikz/grid with coordinates/major step}}
        \foreach \x in {\grd@xa,\grd@xc,...,\grd@xb}
        \node[anchor=north] at (\x,\grd@ya) {\pgfmathprintnumber{\x}};
        \foreach \y in {\grd@ya,\grd@yc,...,\grd@yb}
        \node[anchor=east] at (\grd@xa,\y) {\pgfmathprintnumber{\y}};
      }
    }
  },
  minor help lines/.style={
    help lines,
    step=\pgfkeysvalueof{/tikz/grid with coordinates/minor step}
  },
  major help lines/.style={
    help lines,
    line width=\pgfkeysvalueof{/tikz/grid with coordinates/major line width},
    step=\pgfkeysvalueof{/tikz/grid with coordinates/major step}
  },
  grid with coordinates/.cd,
  minor step/.initial=.2,
  major step/.initial=1,
  major line width/.initial=2pt,
}
\makeatother

\usepackage{thmtools}
\usepackage{thm-restate}
\usepackage{etoolbox}
\makeatletter
\def\problem@s{}
\newcounter{problems@cnt}

\newcommand{\allproblems}{\problem@s}
\makeatother

\usepackage{amsmath,amsfonts,amssymb,array,graphicx,mathtools,multirow,bm,tcolorbox,relsize,booktabs}
\usepackage[shortlabels]{enumitem}
\usepackage{graphicx,epic,times,eepic,epsfig,latexsym,verbatim,color}
\usepackage{nicefrac}       % compact symbols for 1/2, etc.
\usepackage[linesnumbered,ruled,procnumbered]{algorithm2e}
\usepackage{mathrsfs}
\usepackage{framed}
\usepackage{bbm}
\usepackage{epsfig}
\usepackage{pgfplots}
\usepackage[strict]{changepage}
\usepackage{hyperref}
\hypersetup{colorlinks=true,citecolor=blue,linkcolor=blue,filecolor=blue,urlcolor=blue,breaklinks=true}
\usepackage[marginal]{footmisc}
\usepackage{url}
\usepackage{theorem}
\pgfplotsset{compat=1.18} 

\usepackage{thmtools}
\usepackage{thm-restate}
\usepackage{etoolbox}
\usepackage[utf8]{inputenc}
\usepackage[T1]{fontenc}
\usepackage{qcircuit}

\allowdisplaybreaks

\begin{document}
\title{Computable and Faithful Lower Bound on Entanglement Cost}

\author{Xin Wang}
\email{felixxinwang@hkust-gz.edu.cn}

\author{Mingrui Jing}
\email{mjing638@connect.hkust-gz.edu.cn}

\author{Chengkai Zhu}
\email{czhu696@connect.hkust-gz.edu.cn}
\affiliation{Thrust of Artificial Intelligence, Information Hub,\\
The Hong Kong University of Science and Technology (Guangzhou), Guangzhou 511453, China}

\begin{abstract}
Quantifying the minimum entanglement needed to prepare quantum states and implement quantum processes is a key challenge in quantum information theory. In this work, we develop computable and faithful lower bounds on the entanglement cost under quantum operations that completely preserve the positivity of partial transpose (PPT operations), by introducing the generalized divergence of $k$-negativity, a generalization of logarithmic negativity. Our bounds are efficiently computable via semidefinite programming and provide non-trivial values for all states that are non-PPT (NPT), establishing their faithfulness for the resource theory of NPT entanglement. Notably, we find and affirm the irreversibility of asymptotic entanglement manipulation under PPT operations for full-rank entangled states. Furthermore, we extend our methodology to derive lower bounds on the entanglement cost of both point-to-point and bipartite quantum channels. Our bound demonstrates improvements over previously known computable bounds for a wide range of quantum states and channels. These findings push the boundaries of understanding the structure of entanglement and the fundamental limits of entanglement manipulation.
\end{abstract}

\maketitle
%%%%%%%%%%%%%%%%%%%%%%%%%%%%%%%%%%%%%%%%%%%%%%%%%%%%%%%%%%%%%%%%%%%%%%%%%%%
\emph{Introduction.}---
Quantum entanglement, a phenomenon at the heart of quantum mechanics, has been widely recognized as a valuable resource for performing tasks impossible with classical systems alone~\cite{Horodecki2009}. The correlation between entangled quantum systems has led to various applications in quantum computation~\cite{Brus2011,Jozsa2003a}, communication~\cite{Bennett1999}, sensing~\cite{Degen2017}, and cryptography~\cite{Ekert1991}. Developing a comprehensive theory of entanglement quantification is a priority in quantum information~\cite{Wilde2017book,Watrous2011b,Hayashi2017b} and quantum resource theory~\cite{Chitambar2018}.

To harness the full potential of quantum entanglement in various applications, one crucial aspect is quantifying the minimum amount of entanglement needed to prepare a given target state~\cite{Bennett1996b,Hayden2001} or quantum operation~\cite{Berta2015e}. This plays a critical role in optimizing the efficiency of quantum computation and communication. For example, the quantum teleportation protocol ~\cite{Bennett1993} uses one ebit (entanglement bit), measured by the maximally entangled state (MES)~\cite{Gisin1998bell}, to simulate an identity channel that perfectly transfers any quantum state between two parties. Further, efficiently estimating entanglement cost is extremely important for understanding the quantumness of implementing quantum algorithms~\cite{Chen2022much}, performing quantum error correction~\cite{Brun2006correcting}, enhancing quantum sensing~\cite{Xia2023entanglement}, and simulating quantum information processing~\cite{Wang2020c}. Therefore, a computable and faithful estimation method is a commonly pursued endeavor in the field of entanglement resource theory.

\begin{figure}[t]
    \centering
    \includegraphics[width=1.0\linewidth]{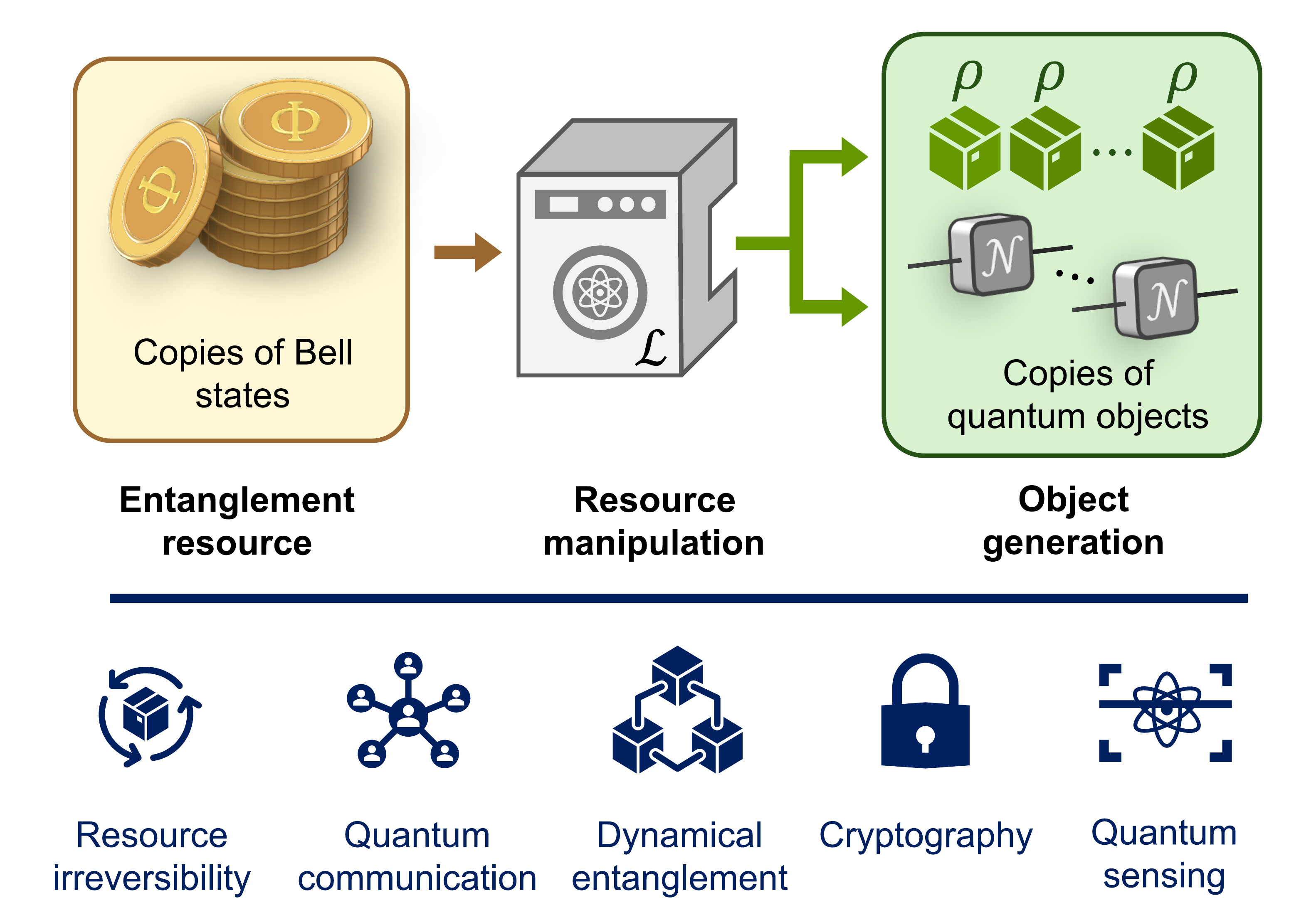}
    \caption{General framework of entanglement cost. Given access to copies of ebits, these golden `coins' are utilized to realize a target quantum state $\rho$ or quantum channel $\cN$ through a process that cannot generate any maximally entangled states (MES) itself. The minimal rate of MES required to achieve this process perfectly is called the entanglement cost of the specific object. This concept plays a crucial role in quantifying entanglement in the studies of resource irreversibility, quantum communication, dynamical entanglement, cryptography, and quantum sensing.
    }
    \label{fig:fig_ent_dist_dilu}
\end{figure}

For the case of state preparation, despite the known challenges in evaluating the entanglement cost via regularizing the entanglement of formation~\cite{Wootters1998entanglement,Winter2016tight,Hastings2009}, numerous attempts have yielded to get valuable bounds. The squashed entanglement~\cite{Christandl2004} is a faithful lower bound, yet it is hardly computable. Two other semidefinite programming (SDP) lower bounds~\cite{Wang2016d,Lami2023a} are not faithful, thus can not estimate the entanglement cost of general quantum states precisely. Only very few instances of states were nearly solved due to their special symmetries~\cite{Audenaert2002}. Beyond that, little was known about estimating the entanglement cost of dynamic quantum processes, which requires both efficient and accurate solutions to the static entanglement cost. Very few quantum channels were investigated with their dynamical cost estimations~\cite{Wilde2018,Gour2020}. The faithful bound of the cost to implement general quantum channels or even bipartite channels can be extremely difficult~\cite{Berta2015e,Wilde2018}.

In this paper, we develop the first efficiently computable and faithful lower bound on the entanglement cost under quantum operations that completely preserve the positivity of partial transpose (PPT operations)~\cite{Audenaert2003,Moor2008,Wang2016d}. Technically, we introduce a sequential classification of quantum sub-states called the $k$-hierarchy of positive partial transpose ($\PPT_k$) and combine it with the quantum R{\'{e}}nyi divergence to arrive at an entanglement quantifier for any bipartite state. This is the first SDP lower bound on the entanglement cost that is also faithful for the NPT entanglement theory. The advantages of our bounds are shown by comparing them with previously known ones on random states. Notably, we demonstrate that entanglement manipulation could be irreversible for fully non-degenerate states asymptotically, a phenomenon not previously observed. We further apply our methodology to the dynamical entanglement cost of both point-to-point and bipartite quantum channels, improving the previous lower bound on the entanglement cost of the Werner-Holevo channel for specific dimensions and providing computable lower bounds on the entanglement cost for more general bipartite quantum channels such as noisy Hamiltonian evolutions.

%%%%%%%%%%%%%%%%%%%%%%%%%%%%%%%%%%%%%%%%%%%%%%%%%%%%%%%%%%%%%%%%%%%%%%%%%%%%%%%%%%%%%%%%%%%%%%%%%%%%%%%%%%%%%%%%%%%%%%%%%%%%%%%%%%%%%%%%%%%%%%%%%%%%%%%%%%%%%%%%%%%%%%%%%%%%%%%%%%%%%%%%%%%%%%%%%%%%%%%%%%%%%%%%%%%%%%%%%%%%
\vspace{2mm}
\emph{Static entanglement cost.}---
Let $\cH$ be a finite-dimensional Hilbert space. Consider two parties, Alice and Bob, with associated Hilbert spaces $\cH_A$ and $\cH_B$, respectively, where the dimensions of $\cH_A$ and $\cH_B$ are denoted as $d_A$ and $d_B$. We denote the set of all linear operators on system $A$ as $\mathscr{L}(\cH_A)$ and the set of all density operators as $\mathscr{D}(\cH_A)$. The maximally entangled state of dimension $d\otimes d$ is denoted as $\Phi_{AB}^+(d)=1/d\sum_{i,j=0}^{d-1} \ketbra{ii}{jj}$. The trace norm of $\rho$ is denoted as $\|\rho\|_1 = \tr(\sqrt{\rho^\dagger \rho})$. 
A bipartite quantum state $\rho_{AB}\in \mathscr{D}(\cH_A\ox\cH_B)$ is called a PPT state if it admits positive partial transpose, i.e., $\rho_{AB}^{T_B}\geq 0$ where $T_B$ denotes taking partial transpose over the system $B$. Note that any state that is not PPT is called NPT, which will be a resource in the resource theory of NPT entanglement~\cite{Gour2020}. A quantum channel $\cN_{A\rightarrow B}$ is a linear transformation from $\mathscr{L}(\cH_A)$ to $\mathscr{L}(\cH_B)$ that is completely positive and trace preserving (CPTP).

In the task of entanglement dilution, Alice and Bob share a large supply of Bell pairs and try to convert $rn$ Bell pairs to $n$ high fidelity copies of the desired state $\rho_{AB}^{\ox n}$ using suitable operations. The \emph{entanglement cost} $E_{C,\Omega}$ of a given bipartite state $\rho_{AB}$ quantifies the optimal rate $r$ of converting ${rn}$  Bell pairs to $\rho_{AB}^{\ox n}$ with arbitrarily high fidelity in the limit of large $n$. The concise definition of the entanglement cost using $\Omega$ operations is given as follows:
\begin{equation*}
E_{C,\Omega}(\rho_{AB})=\inf\{r: \lim_{n \to \infty} \inf_{\Lambda\in \Omega}  \|\rho_{AB}^{\ox n}-\Lambda (\Phi^+_{AB}(2^{rn}))\|_1=0\},
\end{equation*}
where $\Omega\in\{\OLOCC,\LOCC,\SEP,\PPT\}$ and we write $E_{C,\LOCC}(\cdot)=E_C(\cdot)$ for simplification. For the case with an exact transformation rather than the vanishing error, the rate refers to the exact entanglement cost~\cite{Audenaert2003,Wang2020c}.

To advance our understanding of static entanglement, we draw inspiration from the Rains set $\PPT' = \{\sigma_{AB}: \sigma_{AB} \succeq 0, \|\sigma_{AB}^{T_B}\|_1 \leq 1 \}$, initially proposed to estimate the distillable entanglement~\cite{Rains1999a} (see, e.g.,~\cite{Christandl2004,Leditzky2017,Wang2016m,Wang2016c,Fang2017,Kaur2018} for other methods). We generalize the set to a sequence of sub-state sets $\Freek (k\geq 2)$ as
\begin{equation*}
\begin{aligned}
    \Freek \coloneqq& \Big\{\omega_1 \in\mathscr{L}(\cH_{AB}): \omega_1\succeq 0,~\exists \{\omega_i\}_{i=2}^{k},~\mathrm{s.t.}\\
    &|\omega_i^{T_B}|_* \preceq \omega_{i+1}, \forall i\in [1:k-1],~\big\|\omega_{k}^{T_B}\big\|_1\le 1\Big\},
\end{aligned}
\end{equation*}
where $|X|_{*} \preceq Y$ denotes $-Y \preceq X\preceq Y$.
In particular, $\PPT_1$ reduces to the Rains set $\PPT'$. {Moreover, $\PPT_k$ is related to the quantity $\chi_p$ developed in~\cite[Eq.~(6)]{Lami2025} for the computable zero-error PPT entanglement cost, i.e., $\PPT_k = \{\omega_{AB} \succeq 0: \chi_{k-1}(\omega_{AB}) \leq 1\}$.} By construction, any sub-state in $\PPT_{k+1}$ can be shown as an element in $\PPT_{k}$, and the Rains set $\PPT'$ serves as the largest boundary set within this hierarchical structure,
{i.e., $ \PPT\subseteq \PPT_k \subseteq \cdots \subseteq \PPT_2 \subseteq \PPT'$.}
Based on the hierarchy, we introduce the \textit{logarithmic fidelity of binegativity} of any bipartite quantum state $\rho_{AB}$ as follows.
\begin{equation}
   E_{\NB,2}^{1/2}(\rho_{AB}) = - \log \max_{\sigma_{AB}\in\Freesec} F(\rho_{AB},\sigma_{AB}),
\end{equation}
where $F(\rho, \sigma)$ is the fidelity between $\rho$ and $\sigma$. We call it logarithmic fidelity of binegativity because $\PPT_2$ has a similar spirit of binegativity by recursively taking the partial transpose and finally utilizing the negativity (see Appendix~\ref{appendix:properties_renyi_divergence_k_negativity} for the generalized divergence of $k$-negativity and our alternative SDP bound). Our first main result is that the logarithmic fidelity of binegativity actually gives a faithful and efficiently computable lower bound on the entanglement cost under PPT operations.
\begin{theorem}[Lower bound on the static entanglement cost]\label{thm:min_thm}
For any bipartite state $\rho_{AB}\in \mathscr{D}(\cH_{AB})$, its entanglement cost is lower bounded by
\begin{equation}
    E_{C}(\rho_{AB}) \geq E_{C,\PPT}(\rho_{AB}) \geq E_{\NB, 2}^{1/2}(\rho_{AB}).
\end{equation}
\end{theorem}
\textit{Sketch of proof.} Since $\alpha \mapsto D_\alpha (\rho_{AB}||\sigma_{AB})$ for given $\rho_{AB}$ and $\sigma_{AB}$ is a monotonically non-decreasing function for {the sandwiched R{\'{e}}nyi relative entropy} in $\alpha$~\cite{Muller_Lennert2013,Wilde2014a}, 
we have
\begin{equation*}
    D(\rho_{AB}||\sigma_{AB}) \geq D_{1/2}(\rho_{AB}||\sigma_{AB}) = -\log F(\rho_{AB}, \sigma_{AB}).
\end{equation*}
Notice that $\PPT\subsetneq \Freesec$, we have 
\begin{equation*}
    E_{R,\PPT}(\rho_{AB}) \geq - \max_{\sigma_{AB}\in\PPT} \log F(\rho_{AB},\sigma_{AB}) \geq E_{\NB,2}^{1/2}(\rho_{AB}),
\end{equation*}
where $E_{R,\PPT}(\rho_{AB})$ is the PPT relative entropy of entanglement~\cite{Audenaert2002,Plenio2007}. Consequently, we have
\begin{equation}\label{Eq:EFPPT_onecopy}
\begin{aligned}
    E_{C,\PPT}(\rho_{AB}) &\geq \lim_{n\rightarrow \infty} \frac{1}{n}E_{R}(\rho_{AB}^{\ox n})\\
    &\geq \lim_{n\rightarrow \infty} \frac{1}{n} E_{\NB,2}^{1/2}(\rho_{AB}^{\ox n}) = E_{\NB,2}^{1/2}(\rho_{AB}),
\end{aligned}
\end{equation}
where the first inequality follows from Ref.~\cite[p. 421]{Hayashi2006a} and the equality is a consequence of the additivity of $E_{\NB,2}^{1/2}(\cdot)$ proved in Appendix~\ref{appendix:properties_renyi_divergence_k_negativity}.

To the best of our knowledge, the entanglement cost lower bound given in Theorem~\ref{thm:min_thm} is the first faithful lower bound in the frame of NPT entanglement theory, and can also be efficiently computed via SDP. Notably, it establishes a vital connection between the entanglement cost of a bipartite state $\rho_{AB}$ and its distance to the sub-state set $\PPT_2$. In essence, states that exhibit a considerable separation from any sub-state with negative logarithmic binegativity are associated with a correspondingly higher entanglement cost. More details and proofs can be found in Appendix~\ref{appendix:properties_renyi_divergence_k_negativity}.

We compared our bound with previous known ones such as the bound $E_{\eta}$ introduced in Ref.~\cite{Wang2016d}, and the \textit{tempered negativity} $E_{N}^\tau$~\cite{Lami2023a}. We generate random bipartite states of fixed rank with respect to the Hilbert-Schmidt measure~\cite{_yczkowski_2011} and compare the corresponding values of $E_{\eta}$ and $E_{N}^\tau$ with $E_{\NB,2}^{1/2}$. In Fig.~\ref{fig:fixrank_sampling}, it can be observed that higher-rank states tend to have diminishing values of $E_{\eta}$ and $E_{N}^\tau$, and $E_{\NB,2}^{1/2}$ could provider tighter estimation in most of the cases. More comparisons among our bound and existing bounds have been illustrated in Appendix~\ref{appendix:comp_to_other_computable_bounds}.

\begin{figure}
    \centering
    \includegraphics[width=1\linewidth]{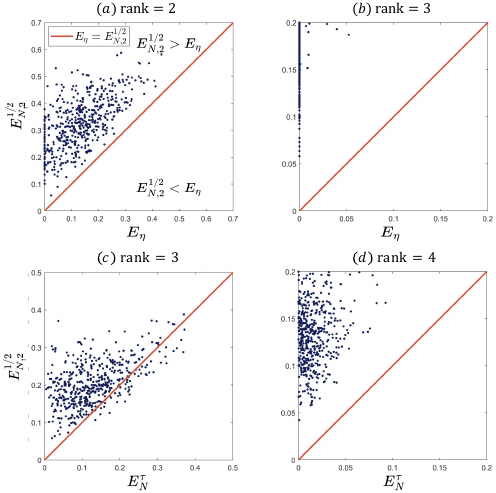}
    \caption{(a-b) Compare $E_{\NB,2}^{1/2}$ with $E_{\eta}$ and (c-d) Compare $E_{\NB,2}^{1/2}$ with $E_{N}^\tau$. Each dot corresponds to one of $500$ randomly generated states according to the Hilbert-Schmidt measure. The red line indicates states $\rho_{AB}$ for which two measures give the same bound value. The dots above the red line (resp. below) indicate the states for which $E_{\NB,2}^{1/2}$ is tighter (resp. $E_{\NB,2}^{1/2}$ is looser).}
    \label{fig:fixrank_sampling}
\end{figure}

Our analysis particularly highlights the effectiveness of the bound $E_{\NB,2}^{1/2}$ in evaluating the entanglement cost of general noisy states, which are not adequately addressed by $E_{\eta}$. For any full-rank, NPT bipartite state $\rho_{AB}$, our findings show that $E_{\NB,2}^{1/2}(\rho_{AB})$ is always greater than $E_{\eta}(\rho_{AB})$, primarily because a full-rank state make $E_{\eta}(\rho_{AB})$ exactly zero. This also demonstrates the faithfulness of our new bounds $E_{\NB, k}$, as they remain positive for all full-rank NPT states.

As a notable application, in the following, we apply our bounds to affirm the irreversibility of asymptotic entanglement manipulation under PPT operations for full-rank quantum states. The reversibility of entanglement manipulation has been an essential topic in quantum information theory. The asymptotic entanglement manipulation is reversible for pure states~\cite{Bennett1996b}, which makes the entanglement entropy a unique entanglement measure of pure states~\cite{Plenio2007,Horodecki2000limits}. However, the manipulation is inherently irreversible for general mixed states under LOCC or PPT operations~\cite{Vidal2001,Vidal2002b,Vollbrecht2004,Cornelio2011,Yang2005,Wang2016d,Lami2023a}, unlike the conservation of resources in thermodynamics. This irreversibility highlights the impossibility of establishing a single measure that governs all entanglement transformations, and establishing reversibility requires a deeper understanding of entanglement manipulation. Whether there is a reversible entanglement theory under some larger class of operations, e.g., asymptotically entanglement non-generating operations, has been intensively studied recently~\cite{Fang2021,Berta2023,hayashi2024,lami2024a}.

To investigate the irreversibility of entanglement manipulation, we introduce the PPT-undistillable entanglement $E_{\text{U,PPT}}(\rho_{AB})$ of a bipartite quantum state as the difference between the entanglement cost and the distillable entanglement of the state regarding the PPT operations. A positive PPT-undistillable entanglement indicates the irreversibility of the entanglement manipulation under PPT operations. This irreversibility has been observed in rank-two states supporting supporting the $3\ox 3$ antisymmetric subspace~\cite{Wang2016d}, exemplified by states like $\rho_v = \frac{1}{2}(\ketbra{v_1}{v_1} + \ketbra{v_2}{v_2})$ where
\begin{equation}
    \ket{v_1} = 1/\sqrt{2}(\ket{01} - \ket{10}), \ \ket{v_2} = 1/\sqrt{2}(\ket{02} - \ket{20}).
\end{equation}
These states are pivotal in quantum information theory, aiding in the exploration of differences between LOCC and separable operations, such as in quantum state discrimination~\cite{Cohen2007} and entanglement transformation~\cite{Chitambar2009}.

The extent to which state makes $E_{\text{U,PPT}}(\rho_{AB})$ positive remains questioned. Such states are crucial to understanding the equilibrium states of quantum many-body systems, which are typically non-degenerate. We consider the case where  $\rho_{v}$ is affected by the depolarizing noise, resulting in a state $\widehat{\rho}_v = (1-p)\rho_v + p I/3$, where $p$ is the depolarizing rate. Our comparative analysis of the logarithmic fidelity of binegativity and the Rains bound~\cite{Rains2001,Audenaert2002} in low-noise scenarios reveals a discernible gap between these metrics (see Fig.~\ref{fig:irreversible_qudit}). This study corroborates previous findings and marks the irreversibility of entanglement for the first time in full-rank entangled states under PPT operations.

\begin{figure}[t]
    \centering
    \includegraphics[width=0.8\linewidth]{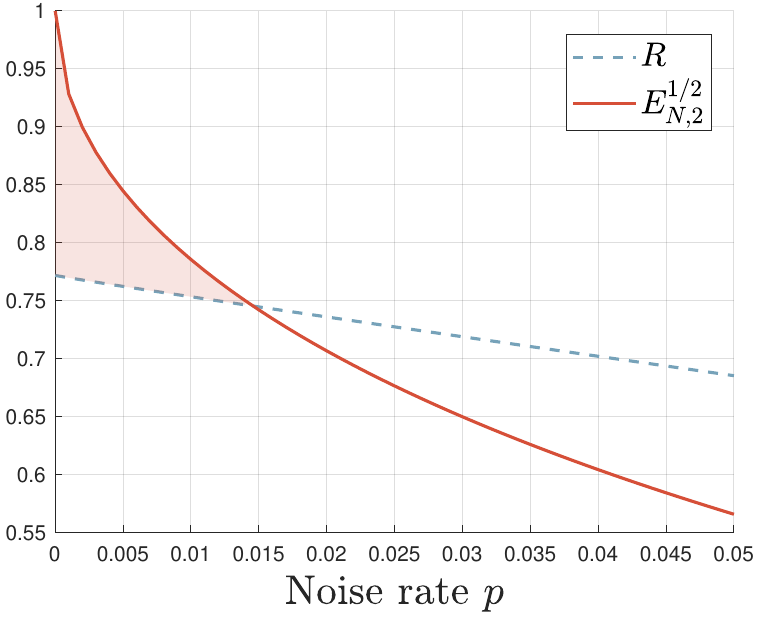}
    \caption{Comparison between $E_{\NB,2}^{1/2}$ and the Rains bound for a two-qutrit full-rank state $\widehat{\rho}_v$. The $x$-axis represents the depolarizing noise parameter $p$. The red line represents $E_{\NB,2}^{1/2}(\widehat{\rho}_v)$, serving as a lower bound on the entanglement cost, while the dashed blue line represents the Rains bound, providing an upper bound on the distillable entanglement. The gap between the two bounds for noise parameter $p \in [0, 0.015]$ indicates the irreversibility of asymptotic entanglement manipulation of $\widehat{\rho}_v$.}
    \label{fig:irreversible_qudit}
\end{figure}

%%%%%%%%%%%%%%%%%%%%%%%%%%%%%%%%%%%%%%%%%%%%%%%%%%%%%%%%%%%%%%%%
%%%%%%%%%%%%%%%%%%%%%%%%%%%%%%%%%%%%%%%%%%%%%%%%%%%%%%%%%%%%%%%%
\vspace{2mm}
\emph{Dynamical entanglement cost.}---
We further extend our theory of static entanglement to the dynamic realms of quantum channels, recognizing their critical role in quantum communication and error correction~\cite{Wilde2018}. The concept of distillable entanglement for point-to-point channels (quantum channels with single input and output systems) was initially explored in Ref.~\cite{Bennett1996c}, and has evolved through continued research~\cite{Devetak2003a,Takeoka2014b}, laying the groundwork for understanding entanglement's operational utility. Furthermore, the entanglement cost of quantum channels~\cite{Berta2015e}, quantifies the minimum ebits required to simulate multiple channel uses in the presence of free classical communication. Followed by Theorem~\ref{thm:min_thm} and the result in~\cite[Corollary 17]{Berta2015e}, we have the following lower bound on the entanglement cost of a point-to-point quantum channel.
\begin{proposition}\label{thm:EntCost_from_choi_E_new}
For a point-to-point quantum channel $\cN_{A\rightarrow B}$, its entanglement cost is lower bounded by
\begin{equation}
     E_C\left(\cN_{A \rightarrow B}\right) \geq \widehat{E}_{\NB,2}^{1/2}(\cN_{A\rightarrow B}) \geq E_{\NB,2}^{1/2}(J_{\cN}),
\end{equation}
where $\widehat{E}_{\NB,2}^{1/2}(\cN_{A\rightarrow B}) \coloneqq \max\limits_{\psi_{A A'}} E_{\NB,2}^{1/2}((\cN_{A \rightarrow B} \otimes \cI_{A^{\prime}})(\psi_{A A'}))$ and $J_{\cN}$ is the Choi state of $\cN_{A\rightarrow B}$.
\end{proposition}

Proposition~\ref{thm:EntCost_from_choi_E_new} provides an efficiently computable lower bound $E_{\NB,2}^{1/2}(J_{\cN})$ on the entanglement cost of a point-to-point quantum channel $\cN_{A\rightarrow B}$. Note that a channel is PPT if and only if its Choi state is a PPT state. Thus, $E_{\NB,2}^{1/2}(J_{\cN})$ is also a faithful bound for $\cN_{A \rightarrow B}$, i.e., $E_{\NB,2}^{1/2}(J_{\cN}) =0$ if and only if $\cN_{A \rightarrow B}$ is a PPT channel. This lower bound indicates the relationship between the asymptotic cost of Bell pairs to simulate a quantum channel $\cN_{A\rightarrow B}$ and the distance of its Choi state $J_{\cN}$ to any sub-state with non-positive logarithmic binegativity. If $J_{\cN}$ is far from any of such sub-states, the asymptotic cost of Bell pairs for the simulation of $\cN_{A\rightarrow B}$ increases.

% \begin{figure}[h!]
%     \centering
%     \includegraphics[width=0.88\linewidth]{Paper_figure/EntCost_bichannel.pdf}
%     \caption{General framework of defining entanglement cost of bipartite channels in terms of channel simulation. The figure shows a simulation of the target bipartite channel $\cN_{AB\rightarrow A'B'}$, up to an $\varepsilon$-error faithfulness, using an LOCC channel $\cL_{AB\Bar{A}\Bar{B}\rightarrow A'B'}$ with $k$ resource states injected.}
%     \label{fig:ent_cost_bichannel}
% \end{figure}

Recent studies have expanded the entanglement theory to bipartite channels where a variety of measures for dynamic entanglement have significantly enriched the scope of entanglement theory in quantum dynamics~\cite{Wang2018semidefinite,Gour2019,Gour2020,Gour2021}. Consider a bipartite quantum channel $\cN_{AB\rightarrow A'B'}$. There are two main strategies to define the dynamical entanglement cost of such channels, i.e., the parallel and adaptive simulation~\cite{Berta2015e,bauml2019resource,Wilde2018,Gour2021}. In the parallel approach, it is defined as the minimal rate $r$ to simulate $n\gg1$ simultaneous uses of $\cN$ with the LOCC channel by consuming $rn$ ebits, up to an error of $\varepsilon$ based on the diamond norm~\cite{Kitaev1997QuantumCA}. Formally, the entanglement cost of simulating one use of $\cN$ with error $\varepsilon$ can be expressed as,
\begin{equation}\label{Eq:channel_entcost_def}
\begin{aligned}
    E_{C,\varepsilon}^{(1)}(\cN) :=& \min\Big\{\log k \, :\, \Big\|\cN_{AB\rightarrow A'B'} - \\
    &\cL_{A B \Bar{A}\Bar{B} \rightarrow A' B'}\big(\cdot \ox \Phi_{\Bar{A}\Bar{B}}^+(2^k)\big)\Big\|_{\diamond}\leq \varepsilon\Big\},
\end{aligned}
\end{equation}
where $k\in \mathbb{N}$ and the minimization ranges over all LOCC operations $\cL_{A B \Bar{A}\Bar{B}}$ between Alice and Bob. The asymptotic entanglement cost of $\cN_{AB\rightarrow A'B'}$ is then defined as~\cite{Gour2020}, 
\begin{equation}
    E_C(\cN) \coloneqq \lim_{\varepsilon \rightarrow 0} \liminf_{n\rightarrow \infty} \frac{1}{n}E_{C, \varepsilon}^{(1)}(\cN^{\ox n}).
\end{equation}
This definition builds upon the entanglement cost for point-to-point channels~\cite{Berta2015e} by taking subsystems $B$ and $A'$ to be trivially one-dimension. 

In our study, we primarily explore the parallel framework for assessing entanglement costs (see adaptive version in~\cite{bauml2019resource,Gour2021}). Our methodology on static entanglement can be naturally generalized to bipartite quantum channels by establishing the relationship between the dynamical entanglement cost of $\cN_{AB\rightarrow A'B'}$ and the static entanglement cost of its Choi state relying on the  NPT entanglement theory.

\begin{proposition}\label{thm:EntCost_from_bipartite_choi_E_new}
For a bipartite quantum channel $\cN_{AB\rightarrow A'B'}$, its entanglement cost is lower bounded by 
\begin{equation}
     E_C\left(\cN_{AB \rightarrow A'B'}\right) \geq E_{\NB,2}^{1/2}(J^{\cN}_{AA'BB'}),
\end{equation}
where $J^{\cN}_{AA'BB'}$ is the Choi state of $\cN_{AB\rightarrow A'B'}$ with the bipartite partition $AA':BB'$.
\end{proposition}
The proof is based on the fact that $E_{C}(\cN_{AB\rightarrow A'B'}) \geq E_C(J_{AA'BB'}^{\cN})$, where $J^{\cN}_{AA'BB'}$ is the Choi state of $\cN_{AB\rightarrow A'B'}$ and the bipartite cut is settled between systems $AA'$ and $BB'$ for calculating $E_C(J_{AA'BB'}^{\cN})$. More details can be found in Appendix~\ref{appendix:bound_bi_channel}.

%%%%%%%%%%%%%%%%%%%%%%%%%%%%%%%%%%%%%%%%%%%%%%%%%%%%%%%%%%%%%%%%%%%%%%%%%%%
%%%%%%%%%%%%%%%%%%%%%%%%%%%%%%%%%%%%%%%%%%%%%%%%%%%%%%%%%%%%%%%%%%%%%%%%%%%
\vspace{2mm}
\emph{Examples.}---
We apply our bounds to different channels of interest. First, we consider the Werner-Holevo channel $\cW_{A\rightarrow B}^{(d)}(\rho)\coloneqq \frac{1}{d-1}(I - \rho^T)$, which provides a counterexample to the multiplicativity of maximal output purity~\cite{Werner2002a} and is the dual of the state that disproves the additivity of the relative entropy of entanglement~\cite{Vollbrecht2001}. The Choi state of $\cW_{A\rightarrow B}^{(d)}$ is $J_{\cW} = \frac{1}{d(d-1)}\big(I_{AB} - d(\Phi^+_{AB})^{T_B}\big)$. Previous lower bounds on the entanglement cost of Werner-Holevo channels rely on the antisymmetric subspace, yielding $E_C(\cW^{(d)})\geq E_C(\alpha_d) \geq \log(4/3)$~\cite{Christandl2012,Wilde2018}. Our bound improves upon this for $d = 5$, with $E_C(J_\cW) \geq E_{\NB,2}^{1/2}(J_{\cW}) > 0.4854 > \log(4/3)$.

Second, we investigate the entanglement cost in a noisy Hamiltonian simulation scenario, which is crucial for understanding the dynamics of quantum systems~\cite{Leimkuhler2004simulating,Ding2024simulating,Clinton2021hamiltonian} and has applications in quantum chemistry and material science~\cite{Berry2015hamiltonian,Low2017optimal,Low2019hamiltonian}. We consider a two-body time-independent Heisenberg-XXZ model with Hamiltonian $H = J_x XX + J_y YY + J_z ZZ$, where $J_x = J_y = -1/2$ and $J_z = -1$, and single-qubit thermal damping in the state $\ket{1}$ at a rate of $0.1$. We estimate the entanglement cost of implementing $e^{-iHt}$ at each small time step using $E_{\NB,2}^{1/2}$ and compare it with other lower bounds. The results are depicted in Fig.~\ref{fig:ent_cost_hevolve}. Our bound provides more precise estimations throughout the evolution, exhibiting a wavy-shaped tendency with a maximum at $t\approx1.57$. This suggests that at least two ebits are required to simulate the evolution operator at any time in the presence of amplitude damping on one of the system qubits.

\begin{figure}[t!]
    \centering
    \includegraphics[width=0.85\linewidth]{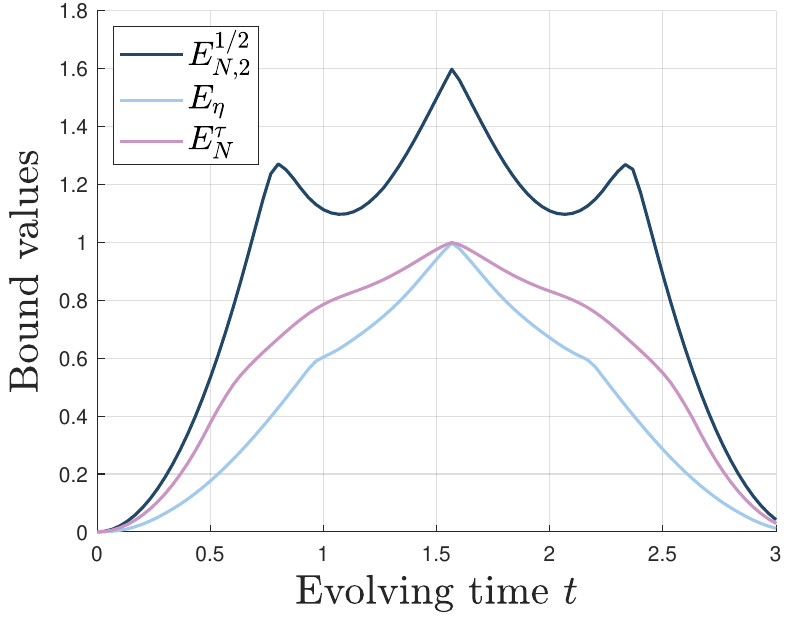}
    \caption{Investigation of the dynamical entanglement during time evolution of a time-independent Hamiltonian. Comparison of the entanglement cost lower bounds $E_{\eta}$~\cite{Wang2016d}, $E^{\tau}_{N}$~\cite{Lami2023a}, and $E^{1/2}_{N,2}$ at each time step of a noisy evolution operator governed by a Heisenberg XXZ Hamiltonian.}
    \label{fig:ent_cost_hevolve}
\end{figure}

%%%%%%%%%%%%%%%%%%%%%%%%%%%%%%%%%%%%%%%%%%%%%%%%%%%%%%%%%%%%%%%%%%%%%%%%%%%
%%%%%%%%%%%%%%%%%%%%%%%%%%%%%%%%%%%%%%%%%%%%%%%%%%%%%%%%%%%%%%%%%%%%%%%%%%%
\emph{Concluding remarks.}---
In this study, we have established efficiently computable lower bounds on the entanglement cost of quantum states and channels under PPT operations, which is faithful for the NPT entanglement theory. The key ingredient in our work is the {generalized divergence of $k$-negativity}, which generalizes the concept of Rains set and logarithmic negativity to a family of sub-state sets and entanglement quantifiers {that are also related to~\cite{Lami2025}}.
Our work settles an important question in quantum information theory by demonstrating the irreversibility of asymptotic entanglement manipulation for full-rank entangled states. Furthermore, our bounds on the entanglement cost of quantum channels pave the way for quantifying the entanglement needed for tasks such as quantum communication and channel simulation. As future directions, it would be interesting to generalize our results beyond the PPT setting, and to explore the applications of logarithmic $k$-negativity in other resource-theoretic frameworks. The proposed bounds on the entanglement cost may also be useful for the recent studies on computational entanglement theory~\cite{Rotem2023,leone2025} and the complexity of distributed quantum computing~\cite{Brenner2023optimal,Piveteau_2024,Harrow2025optimal,Jing2025circuit}.

%%%%%%%%%%%%%%%%%%%%%%%%%%%%%%%%%%%%%%%%%%%%%%%%%%%%%%%%%%%%%%%%%%%%%%%%%
\emph{Acknowledgments}---
The authors would like to thank Ludovico Lami and Chenghong Zhu for their valuable comments. We sincerely thank the program committee of TQC2025 for their very insightful comments that helped us improve the paper.
This work was partially supported by the National Key R\&D Program of China (Grant No.~2024YFE0102500), the National Natural Science Foundation of China (Grant. No.~12447107), the Guangdong Provincial Quantum Science Strategic Initiative (Grant No.~GDZX2403008, GDZX2403001), the Guangdong Provincial Key Lab of Integrated Communication, Sensing and Computation for Ubiquitous Internet of Things (Grant No. 2023B1212010007), the Quantum Science Center of Guangdong-Hong Kong-Macao Greater Bay Area, and the Education Bureau of Guangzhou Municipality.

\emph{Note added}---{The set $\PPT_k$ has been updated since arxiv v3 of this work. The previous $\PPT_k$ in arXiv v2 is now denoted as $\PPT^{\#}_k$ in Appendix~\ref{sec:alternative_lb}, which is also useful for establishing SDP lower bounds for entanglement cost via SDP relaxation. The authors thank TQC2025 PC's helpful comments on the SDP representation of $\PPT^{\#}_k$ and the interesting and explicit connection between the current $\PPT_k$ and $\chi_k$ in \cite{Lami2025}.}

\emph{Data availability}---The data that support the findings of this Letter are openly available~\cite{data_def}.
%%%%%%%%%%%%%%%%%%%%%%%%%%%%%%%%%%%%%%%%%%%%%%%%%%%%%%%%%%%%%%%%%%%%%%%%%
%%%%%%%%%%%%%%%%%%%%%%%%%%%%%%%%%%%%%%%%%%%%%%%%%%%%%%%%%%%%%%%%%%%%%%%%%
% Bibliography

%

\newpage
\numberwithin{equation}{section}
\renewcommand{\theequation}{S\arabic{equation}}
% \numberwithin{equation}{section}
% \renewcommand{\thesubsection}{\normalsize{Supplementary Note \arabic{subsection}}}
\renewcommand{\theproposition}{S\arabic{proposition}}
\renewcommand{\thedefinition}{S\arabic{definition}}
\renewcommand{\thefigure}{S\arabic{figure}}

\vspace{2cm}
\onecolumngrid
\vspace{2cm}

\begin{center}
\large{\textbf{Supplemental Material for} \\ \textbf{
Computable and Faithful Lower Bound on Entanglement Cost}}
\end{center}

%%%%%%%%%%%%%%%%%%%%%%%%%%%%%%%%%%%%%%%%%%%%%%%%%%%%%%%%%%%%%%%%%%%%%%%%%%%%%%%%%%%%
In this Supplemental Material, we offer detailed proofs of the theorems and propositions in the manuscript `Computable and Faithful Lower Bound on Entanglement Cost'. In Section~\ref{appendix:Notations and preliminaries}, we deliver notations and preliminaries on the entanglement theory used in this work. In Section~\ref{appendix:properties_renyi_divergence_k_negativity}, we derive explicit expressions and properties of the generalized $k$-negativity based on the $k$-hierarchy. In Section~\ref{appendix:comp_to_other_computable_bounds}, we compare our bound with previous ones on noisy states. In Section~\ref{appendix:dynamical} and Section~\ref{appendix:bound_bi_channel}, we provide details on the lower bounds on the entanglement cost of point-to-point quantum channels and bipartite quantum channels. In Section~\ref{sec:alternative_lb}, we provide an alternative SDP lower bound on the entanglement cost.

%%%%%%%%%%%%%%%%%%%%%%%%%%%%%%%%%%%%%%%%%%%%%%%%%%%%%%%%%%%%%%%%%%%%%%%%%%%%%%%%%%%%
%%%%%%%%%%%%%%%%%%%%%%%%%%%%%%%%%%%%%%%%%%%%%%%%%%%%%%%%%%%%%%%%%%%%%%%%%%%%%%%%%%%%
\section{Notations and preliminaries}\label{appendix:Notations and preliminaries}

Let $\cH$ be a finite-dimensional Hilbert space. Consider two parties, Alice and Bob, with associated Hilbert spaces $\cH_A$ and $\cH_B$, respectively, where the dimensions of $\cH_A$ and $\cH_B$ are denoted as $d_A$ and $d_B$. We use $\mathscr{L}(\cH_A)$ to represent the set of linear operators on system $A$, and $\mathscr{P}(\cH_A)$ to represent the set of Hermitian and positive semidefinite operators. A linear operator $\rho_A \in \mathscr{P}(\cH_A)$ is considered a density operator of a quantum state if its trace is equal to one. We denote the set of all density operators on system $A$ as $\mathscr{D}(\cH_A)$. The maximally entangled state of dimension $d\otimes d$ is denoted as $\Phi_{AB}^+(d)=1/d\sum_{i,j=0}^{d-1} \ketbra{ii}{jj}$. The trace norm of $\rho$ is denoted as $\|\rho\|_1 = \tr(\sqrt{\rho^\dagger \rho})$. 
A bipartite quantum state $\rho_{AB}\in \mathscr{D}(\cH_A\ox\cH_B)$ is called a PPT state if it admits positive partial transpose, i.e., $\rho_{AB}^{T_B}\geq 0$ where $T_B$ denotes taking partial transpose over the system $B$.
% as $(\ketbra{i_A j_B}{k_A l_B})^{T_B} = \ketbra{i_A l_B}{k_A j_B}$. 
Note that quantum states with non-positive partial transpose are called NPT states, considered resourceful in the NPT entanglement theory~\cite{Gour2021}. A quantum channel $\cN_{A\rightarrow B}$ is a linear transformation from $\mathscr{L}(\cH_A)$ to $\mathscr{L}(\cH_B)$ that is \textit{completely positive and trace preserving} (CPTP).

In the task of entanglement dilution, Alice and Bob share a large supply of Bell pairs and try to convert $rn$ Bell pairs to $n$ high-fidelity copies of the desired state $\rho_{AB}^{\ox n}$ using suitable operations. The \emph{entanglement cost} $E_{C,\Omega}$ of a given bipartite state $\rho_{AB}$ quantifies the optimal rate $r$ of converting ${rn}$  Bell pairs to $\rho_{AB}^{\ox n}$ with arbitrarily high fidelity in the limit of large $n$. The concise definition of the entanglement cost using $\Omega$ operations is given as follows:
\begin{align}
E_{C,\Omega}(\rho_{AB})=\inf\Big\{r: \lim_{n \to \infty}\big[ \inf_{\Lambda\in \Omega}  \big\|\rho_{AB}^{\ox n}-\Lambda (\Phi^+_{AB}(2^{rn}))\big\|_1 \big] = 0\Big\},
\end{align}
where $\Omega \in \{\OLOCC,\LOCC,\SEP,\PPT\}$ and we write $E_{C,\LOCC}(\cdot)=E_C(\cdot)$ for simplification. For the case with an exact transformation rather than the vanishing error, the rate refers to the exact entanglement cost~\cite{Audenaert2003,Wang2020c}.

For the entanglement cost under LOCC, Hayden, Horodecki, and Terhal~\cite{Hayden2001} proved that
\begin{align}
E_{C}(\rho_{AB})=\lim_{n\to \infty}\frac{1}{n}E_F(\rho_{AB}^{\ox n}).
\end{align}
Here $E_F(\rho_{AB}):=\inf \big\{\sum_i p_i S(\tr_A \proj {\psi_i}): \rho_{AB}=\sum_i p_i \proj {\psi_{i}} \big\}$ is the entanglement of formation~\cite{Bennett1996c} of $\rho_{AB}$ where the minimization ranges over all pure state decompositions.
In particular, for any bipartite pure state ${\psi}_{AB}$, it is known that~\cite{Bennett1996b}
\begin{align}\label{Eq:pure_rev}
E_C(\psi_{AB})=E_D(\psi_{AB})=S(\tr_A \psi_{AB}),
\end{align}
where $S(\rho):= - \tr(\rho \log \rho)$ is the \textit{von Neumann entropy} of a state $\rho$. From Eq.~\eqref{Eq:pure_rev} we can see the reversibility in the asymptotic transformation between any pure states.
However, little is known about the entanglement cost of general quantum states due to the computational challenge of mixed states' entanglement measures. One important lower bound on the entanglement cost of $\rho_{AB}$ is the \textit{asymptotic PPT relative entropy of entanglement}~\cite{Hayashi2006a}, 
\begin{equation}
    E_C(\rho_{AB}) \ge \ECPPT(\rho_{AB}) \ge \ERPPTinf(\rho_{AB}):=\inf_{n\ge 1}  \frac{1}{n}\ERPPT(\rho_{AB}^{\ox n}),
\end{equation}
where $\ERPPT(\rho_{AB})$ is  the \textit{PPT relative entropy of entanglement}~\cite{Audenaert2002,Plenio2007} given as follows. 
\begin{align}
    \ERPPT(\rho_{AB}):= \min_{ \sigma_{AB}\in \PPT(A:B)} D(\rho_{AB} \| \sigma_{AB}),
    % &=\min \{D(\rho_{AB} \| \sigma_{AB}): \sigma_{AB} \geq 0, \sigma_{AB}^{T_B}\ge 0, \tr \sigma_{AB}=1\},
\end{align}
where $D(\rho\|\sigma) = \tr\rho(\log \rho-\log \sigma)$ is the quantum relative entropy.

Entanglement distillation, on the other hand, captures the highest rate at which one can obtain maximally entangled states from less entangled states using $\Omega$ operations. It is defined by
\begin{equation}
E_{D}(\rho_{AB})=\sup\Big\{r:\lim_{n \to \infty} \big[\inf_{\Lambda\in\Omega}  \big\|\Lambda(\rho_{AB}^{\ox n})- \Phi_{AB}^+(2^{rn})\big\|_1\big]=0\Big\},
\end{equation}
where $\Omega\in\{\OLOCC,\LOCC,\SEP,\PPT\}$. There are several upper bounds on distillable entanglement~\cite{Rains2001,Christandl2004,Leditzky2017,Wang2016m,Wang2016c,Fang2017,Kaur2018}, and the Rains bound~\cite{Rains2001} is arguably the best-known computable upper bound for the distillable entanglement under two-way LOCC. Its regularization could provide a tighter bound due to its non-additivity~\cite{Wang2016c}. 
Specifically, Rains bound is given by~\cite{Rains2001,Audenaert2002}
\begin{equation}\label{Eq:rains}
    R(\rho_{AB})= \min_{ \tau_{AB}\in \PPT'(A:B)} D(\rho_{AB} \| \tau_{AB}),
\end{equation}
where $\PPT'(A:B)$ is the Rains set defined as
\begin{equation}\label{Eq:rains_sub_state_set}
    \PPT'(A:B) = \big\{\sigma_{AB}: \sigma_{AB} \succeq 0, \big\|\sigma_{AB}^{T_B}\big\|_1 \leq 1 \big\}.
\end{equation}

%%%%%%%%%%%%%%%%%%%%%%%%%%%%%%%%%%%%%%%%%%%%%%%%%%%%%%%%%%%%%%%%%%%%%%%%%%%%%%%%%%%%
%%%%%%%%%%%%%%%%%%%%%%%%%%%%%%%%%%%%%%%%%%%%%%%%%%%%%%%%%%%%%%%%%%%%%%%%%%%%%%%%%%%%
\section{Properties of R\'enyi-$\alpha$ divergence of $k$-negativity}\label{appendix:properties_renyi_divergence_k_negativity}
In the same spirit of the relative entropy of entanglement and the Rains bound, we consider the minimum `distance' between a target state and the sub-state set
\begin{equation}
    \Freek(A:B) \coloneqq \Big\{\omega_1 \in\mathscr{P}(\cH_{AB}):~\exists \{\omega_i\}_{i=2}^{k}, ~\mathrm{s.t.}~ |\omega_i^{T_B}|_* \preceq \omega_{i+1}, \forall i\in [1:k-1],~\big\|\omega_{k}^{T_B}\big\|_1\le 1\Big\},
\end{equation}
where $|A|_{*} \preceq B$ denotes $-B \preceq A\preceq B$. Hence, the sub-state set $\PPT_{k}(A:B)$ can be rewritten as
\begin{equation}
    \Freek(A:B) \coloneqq \Big\{\omega_1 \in\mathscr{P}(\cH_{AB}):~\exists \{\omega_i\}_{i=2}^{k}, ~\mathrm{s.t.}~ -\omega_{i+1} \preceq \omega_i^{T_B} \preceq \omega_{i+1}, \forall i\in [1:k-1],~\big\|\omega_{k}^{T_B}\big\|_1\le 1\Big\}.
\end{equation}
Based on the above definition, we can have a hierarchical structure of sub-state sets encompassing the Rains set as the largest boundary set.
\begin{proposition}\label{prop:hirearchy}
For a bipartite quantum system $AB$ and any positive integer $k\geq 2$ it holds that
\begin{equation}
    \PPT(A:B) \subseteq \PPT_k(A:B) \subseteq \cdots \subseteq \PPT_2(A:B) \subseteq \PPT'(A:B).
\end{equation}
\end{proposition}

\begin{proof}
First, for any fixed $j\in[2:k]$ and any sub-state $\sigma_{AB}\in \PPT_{j}$, we have $-\omega_2\preceq \sigma_{AB}^{T_B} \preceq \omega_2,~-\omega_3\preceq \omega_2^{T_B} \preceq \omega_3, \cdots, -\omega_j \preceq \omega_{j-1}^{T_B} \preceq \omega_{j},\|\omega_{j}^{T_B}\|_1\le 1$. Since $-\omega_j \preceq \omega_{j-1}^{T_B} \preceq \omega_{j}$, by Lemma~\ref{lem:trnormAtotrB}, we have 
\begin{equation}
\big\|\omega_{j-1}^{T_B}\big\|_1 \leq \tr \omega_j.
\end{equation}
Since $\omega_j\succeq 0$, we have
\begin{equation}
    \big\|\omega_{j-1}^{T_B}\big\|_1 \leq \tr\omega_j = \tr\omega_j^{T_B} \leq \big\|\omega_{j}^{T_B}\big\|_1 \leq 1,
\end{equation}
which indicates that $\sigma_{AB}\in \PPT_{j-1}$. Hence, we have shown that $\PPT_j\subseteq\PPT_{j-1}$. Second, if $\sigma_{AB}\in \PPT$, we have $\sigma_{AB}^{T_B}=\omega_1\geq 0$. Then we can construct a sequence of $\omega_{j}$ as $\omega_k = \rho_{AB}$ if $k$ is odd; $\omega_{k} = \omega_1$ if $k$ is even. Therefore, it follows that $\|\omega_{k}^{T_B}\|_1 \leq 1$ and $\sigma_{AB}\in\PPT_k$. Hence, we complete the proof.
\end{proof}

\begin{lemma}\label{lem:trnormAtotrB}
For hermitian matrices $A$ and $B$, if $-B\preceq A\preceq B$, then it holds that $\|A\|_1\leq \tr B$.
\end{lemma}
\begin{proof}
Denote the spectral decomposition of $A$ as $A = \sum_{i=1}^{r} \lambda_{i}\ketbra{\psi_i}{\psi_i}$ and $\{\ket{\psi_i}\}_{i=1}^{d}$ as an orthonormal basis where $r$ and $d$ are the rank and dimension of $A$, respectively. Since $-B \preceq A \preceq B$, we have 
\begin{equation}
    -\bra{\psi} B\ket{\psi} \leq \bra{\psi} A \ket{\psi}\leq \bra{\psi} B\ket{\psi},~\forall \ket{\psi} \implies |\bra{\psi} A \ket{\psi}| \leq \bra{\psi} B\ket{\psi},~\forall \ket{\psi}.
\end{equation}
It follows that 
\begin{equation}
\|A\|_1 = \sum_{i=1}^{r} |\lambda_{i}| = \sum_{i=1}^{d} \big|\bra{\psi_i}A\ket{\psi_i}\big| \leq \sum_{i=1}^{d} \bra{\psi_i}B\ket{\psi_i} = \tr B,
\end{equation}
which completes the proof.
\end{proof}

\begin{lemma}\label{lem:LB_leq0_PPT}
For any positive integer $k\geq 2$, a bipartite quantum state $\sigma_{AB}\in\PPT_k$ if and only if $\sigma_{AB}\in \PPT$.
\end{lemma}
\begin{proof}
The `if' part directly follows from Proposition~\ref{prop:hirearchy}. For the `only if' part, suppose there is a quantum state $\sigma_{AB}\in\PPT_2$ and $\sigma_{AB}\notin \PPT$. Then there exists some $\omega_{AB}\in \mathscr{P}(\cH_{AB})$ such that $-\omega_{AB} \preceq \sigma_{AB}^{T_B} \preceq \omega_{AB}$ and $\|\omega_k^{T_B}\|_1 \leq 1$. Using the same argument in Proposition~\ref{prop:hirearchy}, we have
\begin{equation}
    1 < \|\sigma_{AB}^{T_B}\|_1 \leq \tr \omega_{AB} \leq \|\omega_{AB}^{T_B}\|_1 \leq 1,
\end{equation}
a contradiction. Hence, we have shown that there is no such a state $\sigma_{AB}\in\PPT_2$ and $\sigma_{AB}\notin \PPT$. As $\PPT_k \subseteq \PPT_2$ for any $k\geq 2$, we complete the proof.
\end{proof}

Notably, $\Freek(A:B)$ has a semidefinite programming (SDP) representation, e.g.,
\begin{equation}
    \Freesec(A:B)\coloneqq \Big\{\sigma_{AB}\succeq 0: \exists \omega_{AB}\in\mathscr{L}(\cH_{AB}),~\mathrm{s.t.}~ -\omega_{AB} \preceq \sigma_{AB}^{T_B} \preceq \omega_{AB}, ~\big\|\omega_{AB}^{T_B}\big\|_1\le 1\Big\}.
\end{equation}
We remark that the entanglement cost lower bound $E_{\eta}(\rho_{AB})$ in Ref.~\cite{Wang2016d} can be interpreted as the minimum min-relative entropy between a bipartite state $\rho_{AB}$ and the set $\Freesec(A:B)$.

% the binegativity considered in Ref.~\cite{Moor2008,Ishizaka2004a}. It is proved that $|\sigma_{AB}^{T_B}|^{T_B}\geq 0$ for any pure state~\cite{Moor2008} and any two-qubit quantum state~\cite{Ishizaka2004a}. 

Consider a generalized divergence~\cite{Tomamichel2015b,Khatri2020} $\mathbf{D}(\cdot\|\cdot)$ of a state $\rho$ and a positive semidefinite operator $\sigma$ as a function that obeys: 
\begin{enumerate}
    \item $\mathbf{D}(\rho\|\rho) = 0$ for any state $\rho$;
    \item Logarithmic scaling relation $\mathbf{D}(\rho\|c\sigma) = \mathbf{D}(\rho\|\sigma) - \log c$, for all $c>0$; 
    \item Data processing inequality $\mathbf{D}(\rho\|\sigma) \geq \mathbf{D}(\cN(\rho)\|\cN(\sigma))$ where $\cN$ is an arbitrary quantum channel
\end{enumerate}
We then define an entanglement quantifier called the \textit{generalized divergence of $k$-negativity} of a given bipartite quantum state $\rho_{AB}$ acting on the composite Hilbert space $\cH_{AB}$ as,
\begin{equation}
    E_{\NB, k}(\rho_{AB}) = \min_{\sigma_{AB}\in \Freek(A:B)} \mathbf{D}(\rho_{AB}||\sigma_{AB}).
\end{equation}
One can immediately show that the quantifier is faithful regarding the $\PPT$ states since it is never negative and equals to zero if and only if $\rho_{AB}$ is $\PPT$. 

\begin{lemma}[Faithfulness of $E_{\NB,k}$]\label{lem:faithfulness}
    For a bipartite quantum state $\rho_{AB}\in\mathscr{D}(\cH_A\ox\cH_B)$, $E_{\NB,k}(\rho_{AB})\geq 0$ and $E_{\NB,k}(\rho_{AB}) = 0$ if and only if $\rho_{AB}\in \PPT(A:B)$.
\end{lemma}
\begin{proof}
    $E_{\NB,k}(\rho_{AB})\geq 0$ can be directly obtained by $\mathbf{D}(\rho\|\sigma) \geq 0$ when $\tr \sigma \leq 1$. For the ‘if’ part, since $\PPT\subsetneq\Freesec$, we have $E_{\NB,k}(\rho_{AB}) \le \mathbf{D}(\rho_{AB}||\rho_{AB}) = 0$ which gives $E_{\NB,k}(\rho_{AB}) = 0$. For the ‘only if’ part, since $\mathbf{D}(\rho_{AB}\|\sigma_{AB}) = 0$ if and only if $\rho_{AB} = \sigma_{AB}$, we conclude $\rho_{AB}\in \PPT$ by Lemma~\ref{lem:LB_leq0_PPT}.
\end{proof}

A representative generalized divergence is the \textit{sandwiched R{\'{e}}nyi relative entropy}, given by~\cite{Muller_Lennert2013,Wilde2014a},
\begin{equation}
    \widetilde{D}_\alpha(\rho \| \sigma)\coloneqq\frac{1}{\alpha-1} \log \operatorname{Tr}\left[\left(\sigma^{\frac{1-\alpha}{2 \alpha}} \rho \sigma^{\frac{1-\alpha}{2 \alpha}}\right)^\alpha\right]
\end{equation}
where $\alpha \in(0,1) \cup(1, \infty)$. Then we obtain a family of \textit{R{\'{e}}nyi-$\alpha$ divergence of $k$-negativity} as 
\begin{equation}
    E_{\NB, k}^{\alpha}(\rho_{AB}) = \min_{\sigma_{AB}\in \Freek(A:B)} \widetilde{D}_{\alpha}(\rho_{AB}||\sigma_{AB}).
\end{equation}
We note that $E_{\NB, k}^{\alpha}(\rho_{AB})$ can be efficiently computed via SDP for certain $\alpha$. In particular, when $\alpha = 1/2$, we have the \textit{logarithmic fidelity of binegativity} of any bipartite quantum state $\rho_{AB}$ as follows.
\begin{equation}
   E_{\NB,2}^{1/2}(\rho_{AB}) = - \log \max_{\sigma_{AB}\in\Freesec(A:B)} F(\rho_{AB},\sigma_{AB}),
\end{equation}
where $F(\rho, \sigma)=\left(\tr\left[(\sqrt{\sigma}\rho\sqrt{\sigma})^\frac{1}{2}\right]\right)^2$ is the Uhlmann’s fidelity between $\rho$ and $\sigma$. It has the following properties.

\begin{proposition}[Normalization]\label{prop:normalization}
For a $d\ox d$ maximally entangle state $\Phi^+_{AB}$, $E_{\NB,2}^{1/2}(\Phi^+_{AB}) = \log d$.
\end{proposition}
\begin{proof}
Firstly, we note that the substate $\Phi^+_{AB} / d\in \Freesec(A:B)$, as $\left|(\Phi^+_{AB})^{T_B}\right|^{T_B} = |S_d|^{T_B}/d = I_{AB}/d$,
where we denote $S_{d} = (\Phi_{AB}^+)^{T_B}$. This yields that $1/d = F(\Phi^+_{AB}, \Phi^+_{AB}/d)\leq \widehat{F}_{\NB,2}(\Phi^+_{AB})$. Thus, we have $E_{\NB,2}^{1/2}(\Phi^+) \geq \log d$. Secondly, considering the entanglement cost of pure states, we directly have  $E_{\NB,2}^{1/2}(\Phi^+_{AB}) \leq S(\tr_B \Phi^+_{AB}) = \log d$. In conclusion, we have $E_{\NB,2}^{1/2}(\Phi^+_{AB}) = \log d$.
\end{proof}

\begin{proposition}[Faithfulness of $E_{\NB,2}^{1/2}$]\label{prop:faithful}
For any bipartite state $\rho_{AB}\in\mathscr{D}(\cH_A\ox\cH_B)$, $E_{\NB,2}^{1/2}(\rho_{AB}) \ge 0$. $E_{\NB,2}^{1/2}(\rho_{AB}) = 0$ if and only if $\rho_{AB} \in \PPT(A:B)$. This means the logarithmic fidelity of binegativity is faithful considering the resource theory of NPT entanglement.
\end{proposition}
The faithfulness of $E_{\NB,2}^{1/2}$ can be directly obtained by Lemma~\ref{lem:faithfulness}.

\begin{lemma}
For a bipartite quantum state $\rho_{AB}\in\mathscr{D}(\cH_A\ox\cH_B)$, the logarithmic fidelity of binegativity of $\rho_{AB}$ can be calculated by $E_{\NB,2}^{1/2}(\rho_{AB}) = - 2\log \widehat{F}_{\NB,2}(\rho_{AB})$ where $\widehat{F}_{\NB,2}(\rho_{AB})$ can be computed via the following SDPs. The primal SDP reads
\begin{subequations}
\begin{align}
\widehat{F}_{\NB,2}(\rho_{AB}) = \max & \;\; \frac{1}{2}\tr(X_{AB} + X_{AB}^\dagger)\\
{\rm s.t.}   &\;\; M_{AB},N_{AB}\succeq 0,\\
&\; -M_{AB}^{T_B} + N_{AB}^{T_B} \preceq \sigma^{T_B}_{AB}\preceq M_{AB}^{T_B} - N_{AB}^{T_B}\label{Eq:sig_TB_omega}\\
&\; \tr(M_{AB} + N_{AB}) \leq 1,\label{Eq:omega_TB_leq1}\\
&\; \left(\begin{array}{cc}
\rho_{AB} & X_{AB} \\
X_{AB}^\dagger & \sigma_{AB}
\end{array}\right) \succeq 0,
\end{align}
\end{subequations}
and the dual SDP reads
\begin{equation}\label{sdp:dualSDP}
\begin{aligned}
\widehat{F}_{\NB,2}(\rho_{AB}) = \min & \;\; \tr(Q_{AB}\rho_{AB})\\
{\rm s.t.}  &\;\; U_{AB}\succeq 0,~V_{AB}\succeq 0,~ U_{AB}^{T_B} - V_{AB}^{T_B}\succeq R_{AB},\\
&\;\; -\tr(Q_{AB}\rho_{AB}) I_{AB} \preceq U_{AB}^{T_B} + V_{AB}^{T_B} \preceq \tr(Q_{AB}\rho_{AB}) I_{AB},\\
&\;\; \left(\begin{array}{cc}
Q_{AB} & -I_{AB} \\
-I_{AB} & R_{AB}
\end{array}\right) \succeq 0.
\end{aligned}
\end{equation}
\end{lemma}
\begin{proof}
Recall that $\Freesec(A:B)$ is given by
\begin{equation*}
    \Freesec(A:B)\coloneqq \big\{\sigma_{AB}\succeq 0: \exists\omega_{AB}\in\mathscr{L}(\cH_{AB}),\ \text{s.t.} -\omega_{AB}\preceq \sigma_{AB}^{T_B} \preceq \omega_{AB}, \|\omega_{AB}^{T_B}\|_1 \le 1\big\}.
\end{equation*} 
The logarithmic fidelity of binegativity is given by $E_{\NB,2}^{1/2}(\rho_{AB}) = - 2\log \widehat{F}_{\NB,2}(\rho_{AB})$ where
\begin{equation}\label{eq:EN_1/2}
\begin{aligned}
    \widehat{F}_{\NB,2}(\rho_{AB}) = \max & \;\; \frac{1}{2}\tr(X_{AB} + X_{AB}^\dagger)\\
     {\rm s.t.}   &\;\; -\omega_{AB} \preceq \sigma_{AB}^{T_B} \preceq \omega_{AB},\\
     &\;\;\tr|\omega_{AB}^{T_B}| \leq 1,\\
     &\;\; \left(\begin{array}{cc}
         \rho_{AB} & X_{AB} \\
         X_{AB}^\dagger & \sigma_{AB}
     \end{array}\right) \succeq 0.
\end{aligned}
\end{equation}
It can be equivalently written as
\begin{subequations}\label{sdp:primalSDP}
\begin{align}
    \widehat{F}_{\NB,2}(\rho_{AB}) = \max & \;\; \frac{1}{2}\tr(X_{AB} + X_{AB}^\dagger)\\
     {\rm s.t.}   &\;\; M_{AB},N_{AB}\succeq 0,\\
     &\; -M_{AB}^{T_B} + N_{AB}^{T_B} \preceq \sigma^{T_B}_{AB}\preceq M_{AB}^{T_B} - N_{AB}^{T_B}\\
     &\; \tr(M_{AB} + N_{AB}) \leq 1,\\
     &\; \left(\begin{array}{cc}
         \rho_{AB} & X_{AB} \\
         X_{AB}^\dagger & \sigma_{AB}
     \end{array}\right) \succeq 0.
\end{align}
\end{subequations}
Introducing the Lagrange multipliers $U,V,\left(\begin{array}{cc}
         Q & Y \\
         Y^\dagger & R
    \end{array}\right)$ and $t \in \RR$, the Lagrangian of this primal SDP is given by
\begin{equation}
\begin{aligned}
    &L(Q,R,U,V,Y,t)\\
    = &\; \frac{1}{2}\tr(X + X^\dagger) + \tr\Big[U\big(M^{T_B}-N^{T_B} - \sigma_{AB}^{T_B}\big)\Big] + \tr\Big[V\big(M^{T_B}-N^{T_B} + \sigma_{AB}^{T_B}\big)\Big]\\
    &+ \tr\left(\left(
    \begin{array}{cc}
         Q & Y \\
         Y^\dagger & R
    \end{array}
    \right)\left(
    \begin{array}{cc}
         \rho & X \\
         X^\dagger & \sigma_{AB}
    \end{array}
    \right)\right) + t\Big[1-\tr(M+N)\Big]\\
    = &\; \tr(Q\rho) + t + \tr\Big[\sigma_{AB}(V^{T_B} - U^{T_B} + R)\Big] + \tr\Big[M \big(U^{T_B} + V^{T_B} - tI \big)\Big] + \tr\Big[N\big(-U^{T_B}-V^{T_B} - tI\big)\Big]\\
    &+ \tr\Big[X(Y^\dagger + I/2)\Big] + \tr\Big[X^\dagger(Y + I/2)\Big].
\end{aligned}
\end{equation}
Hence, the dual SDP is given by
\begin{equation}
\begin{aligned}
\min & \;\; \tr(Q_{AB}\rho_{AB}) + t\\
{\rm s.t.}   &\;\; U_{AB}\succeq 0,~V_{AB}\succeq 0,~ U_{AB}^{T_B} - V_{AB}^{T_B}\succeq R_{AB},\\
&\;\; -tI_{AB} \preceq U_{AB}^{T_B} + V_{AB}^{T_B} \preceq tI_{AB},\\
&\;\; \left(\begin{array}{cc}
Q_{AB} & -I_{AB}/2 \\
-I_{AB}/2 & R_{AB}
\end{array}\right) \succeq 0.
\end{aligned}
\end{equation}
Note that for any $\delta_1,\delta_2,\delta_3$ such that $0 \leq \delta_1 < \delta_2,~\delta_1+\delta_2 \leq \frac{1}{d_Ad_B}$ and $\delta_2-\delta_1 > \delta_3>0$, $\{M_{AB} = \delta_2 I_{AB}, N_{AB} = \delta_1 I_{AB},\sigma_{AB} = \delta_3 I_{AB}\}$ is strictly feasible for the SDP~\eqref{sdp:primalSDP}. The Slater condition for strong duality is satisfied. Notice that by rescaling the variables, the dual SDP becomes,
\begin{equation}
\begin{aligned}
\min & \;\; \frac{1}{2}(\tr(Q_{AB}\rho_{AB}) + t)\\
{\rm s.t.}   &\;\; U_{AB}\succeq 0,~V_{AB}\succeq 0,~ U_{AB}^{T_B} - V_{AB}^{T_B}\succeq R_{AB},\\
&\;\; -tI_{AB} \preceq U_{AB}^{T_B} + V_{AB}^{T_B} \preceq tI_{AB},\\
&\;\; \left(\begin{array}{cc}
Q_{AB} & -I_{AB} \\
-I_{AB} & R_{AB}
\end{array}\right) \succeq 0.\label{Eq.dual_problem_schur_constraint}
\end{aligned}
\end{equation}
Notice that the constraint in Eq.~\eqref{Eq.dual_problem_schur_constraint} will lead to $Q_{AB}, R_{AB} \succeq 0$ and $R_{AB}\succeq Q_{AB}^{-1}$. It follows $\tr[Q_{AB}\rho_{AB}] \geq 0$. We also have $t\geq 0$. By the arithmetic-geometric mean inequality, we have
\begin{equation}
\frac{1}{2}\tr[Q_{AB}\rho_{AB}] + \frac{1}{2}t \geq \sqrt{t\cdot \tr[Q_{AB}\rho_{AB}]},    
\end{equation}
where the equality holds iff $\tr[Q_{AB}\rho_{AB}] = t$. In fact, if $\{Q_{AB}, t\}$ is a feasible solution, we can always find a feasible solution $\{\eta Q_{AB}, t/\eta\}$ where $\eta\in\mathbb{R}$ is a scaling factor such that $t =\eta^2 \tr[Q_{AB}\rho_{AB}]$. Also, $\{Q_{AB}, t\}$ and $\{\eta Q_{AB}, t/\eta\}$ will lead to the same value of $\sqrt{t\cdot \tr[Q_{AB}\rho_{AB}]}$. Hence, after restricting the constraint to $\tr[Q_{AB}\rho_{AB}] = t$, we have the following equivalent dual problem.
\begin{equation}
\begin{aligned}
\widehat{F}_{\NB,2}(\rho_{AB}) = \min & \;\; \tr(Q_{AB}\rho_{AB})\\
{\rm s.t.}  &\;\; U_{AB}\succeq 0,~V_{AB}\succeq 0,~ U_{AB}^{T_B} - V_{AB}^{T_B}\succeq R_{AB},\\
&\;\; -\tr(Q_{AB}\rho_{AB}) I_{AB} \preceq U_{AB}^{T_B} + V_{AB}^{T_B} \preceq \tr(Q_{AB}\rho_{AB}) I_{AB},\\
&\;\; \left(\begin{array}{cc}
Q_{AB} & -I_{AB} \\
-I_{AB} & R_{AB}
\end{array}\right) \succeq 0.
\end{aligned}
\end{equation}
\end{proof}

One nice property of the logarithmic fidelity of binegativity is its additivity with respect to the tensor product of states.

\begin{lemma}[Additivity w.r.t. tensor product]\label{lem:additiv}
For any bipartite quantum states $\rho_{0},\rho_1 \in \mathscr{D}(\cH_A\ox \cH_B)$,
\begin{equation}
    E_{\NB,2}^{1/2}(\rho_0\ox \rho_1) = E_{\NB,2}^{1/2}(\rho_0) + E_{\NB,2}^{1/2}(\rho_1).
\end{equation}
\end{lemma}
\begin{proof}
The additivity of $E_{\NB,2}^{1/2}(\cdot)$ can be obtained from the multiplicativity of $\widehat{F}_{\NB,2}(\cdot)$. First, we are going to prove 
\begin{equation}
    \widehat{F}_{\NB,2}(\rho_0\ox\rho_1) \leq \widehat{F}_{\NB,2}(\rho_0)  \widehat{F}_{\NB,2}(\rho_1).
\end{equation}
If we assume the optimal solution to the SDP~\eqref{sdp:dualSDP} for $\widehat{F}_{\NB,2}(\rho_0)$ and $\widehat{F}_{\NB,2}(\rho_1)$ are $\{Q_0,R_0,U_0,V_0\}$ and $\{Q_1,R_1,U_1,V_1\}$, respectively. Then, let
\begin{equation}
\begin{aligned}
    U &= U_0\ox U_1 + V_0 \ox V_1,\\
    V &= U_0\ox V_1 + V_0 \ox U_1.
\end{aligned}
\end{equation}
It follows
\begin{equation}
\begin{aligned}
    (U + V)^{T_{B}} &= (U_0 + V_0)^{T_{B}} \ox (U_1 + V_1)^{T_{B}},\\
    (U - V)^{T_{B}} &= (U_0 - V_0)^{T_B} \ox (U_1 - V_1)^{T_B}.
\end{aligned}
\end{equation}
Then we have
\begin{equation}
    -\tr[(Q_0\ox Q_1)(\rho_0\ox\rho_1)]I \preceq (U + V)^{T_B} \preceq \tr[(Q_0\ox Q_1)(\rho_0\ox\rho_1)]I,
\end{equation}
and $(U - V)^{T_{B}} \geq R_0\ox R_1$. 
Therefore, $\{Q_0\ox Q_1,R_0\ox R_1,U, V\}$ is a feasible solution to the SDP of $\widehat{F}_{\NB,2}(\rho_0\ox\rho_1)$ which implies
\begin{equation}\label{Eq:F_sup_add}
    \widehat{F}_{\NB,2}(\rho_0\ox\rho_1) \leq \tr(Q_0\rho_0)\tr(Q_1\rho_1) = \widehat{F}_{\NB,2}(\rho_0)  \widehat{F}_{\NB,2}(\rho_1).
\end{equation}
Second, we shall show that
\begin{equation*}
    \widehat{F}_{\NB,2}(\rho_0\ox\rho_1) \ge \widehat{F}_{\NB,2}(\rho_0)  \widehat{F}_{\NB,2}(\rho_1).
\end{equation*}
If $\sigma_0,\sigma_1\in \Freesec(A:B)$ are the optimal solutions in Eq.~\eqref{eq:EN_1/2} for $\rho_0$ and $\rho_1$, respectively, then there exists $\omega_0,\omega_1\in\mathscr{P}(\cH_{AB})$ such that
\begin{equation}
\begin{aligned}
    &-\omega_0 \leq \sigma_0^{T_B} \leq \omega_0,~\|\omega_0^{T_B}\|_1\leq 1,\\
    &-\omega_1 \leq \sigma_1^{T_B} \leq \omega_1,~\|\omega_1^{T_B}\|_1\leq 1.
\end{aligned}
\end{equation}
It follows that
\begin{equation}
    (\sigma_0^{T_B}+\omega_0)\ox(\sigma_1^{T_B}+\omega_1) \geq 0,~ (\omega_0-\sigma_0^{T_B})\ox(\omega_1-\sigma_1^{T_B}) \geq 0,
\end{equation}
which gives
\begin{equation}
\begin{aligned}
    &\sigma_0^{T_B}\ox\sigma_1^{T_B} + \omega_0\ox\sigma_1^{T_B} + \sigma_0^{T_B}\ox \omega_1 + \omega_0\ox \omega_1\geq 0,\\
    &\omega_0\ox \omega_1 - \sigma_0^{T_B}\ox\omega_1 - \omega_0\ox\sigma_1^{T_B} + \sigma_0^{T_B}\ox \sigma_1^{T_B} \geq 0.
\end{aligned}
\end{equation}
Summing these two inequalities will yield $\sigma_0\ox\sigma_1 \geq -\omega_0\ox\omega_1$. Meanwhile, consider that
\begin{equation}
    (\sigma_0^{T_B}+\omega_0)\ox(\omega_1-\sigma_1^{T_B}) \geq 0,~(\omega_0-\sigma_0^{T_B}) \ox (\sigma_1^{T_B}+\omega_1) \geq 0,
\end{equation}
which gives
\begin{equation}
\begin{aligned}
    &\sigma_0^{T_B}\ox\omega_1 + \omega_0\ox\omega_1 - \sigma_0^{T_B}\ox \sigma_1^{T_B} - \omega_0\ox \sigma_1^{T_B} \geq 0,\\
    &\omega_0 \ox \sigma_1^{T_B}- \sigma_0^{T_B}\ox \sigma_1^{T_B} +\omega_0\ox\omega_1 - \sigma_0^{T_B} \ox\omega_1\geq 0.
\end{aligned}
\end{equation}
Summing these two inequalities will yield $\sigma_0\ox\sigma_1 \leq \omega_0\ox\omega_1$. Note that $\|\omega_0\ox\omega_1\|_1 \leq \|\omega_0\|_1\|\omega_1\|_1\leq 1$. Therefore, we have $-\omega_0\ox\omega_1\leq \sigma_0\ox\sigma_1 \leq \omega_0\ox\omega_1$ which indicates that $\sigma_0 \ox \sigma_1$ is a feasible solution for $\rho_0 \ox \rho_1$. Therefore, we have
\begin{equation}\label{Eq:F_sub_add}
    \widehat{F}_{\NB,2}(\rho_0\ox\rho_1) \geq \widehat{F}_{\NB,2}(\rho_0)  \widehat{F}_{\NB,2}(\rho_1).
\end{equation}
Combining Eq.~\eqref{Eq:F_sup_add}, Eq.~\eqref{Eq:F_sub_add}, we conclude that 
\begin{equation*}
    \widehat{F}_{\NB,2}(\rho_0\ox\rho_1) = \widehat{F}_{\NB,2}(\rho_0)  \widehat{F}_{\NB,2}(\rho_1),
\end{equation*}
which implies
\begin{equation}
    E_{\NB,2}^{1/2}(\rho_0\ox \rho_1) = E_{\NB,2}^{1/2}(\rho_0) + E_{\NB,2}^{1/2}(\rho_1).
\end{equation}
\end{proof}

%%%%%%%%%%%%%%%%%%%%%%%%%%%%%%%%%%%%%%%%%%%%%%%%%%
%%%%%%%%%%%%%%%%%%%%%%%%%%%%%%%%%%%%%%%%%%%%%%%%%%
\begin{figure}[t]
    \centering
    \includegraphics[width=.75\linewidth]{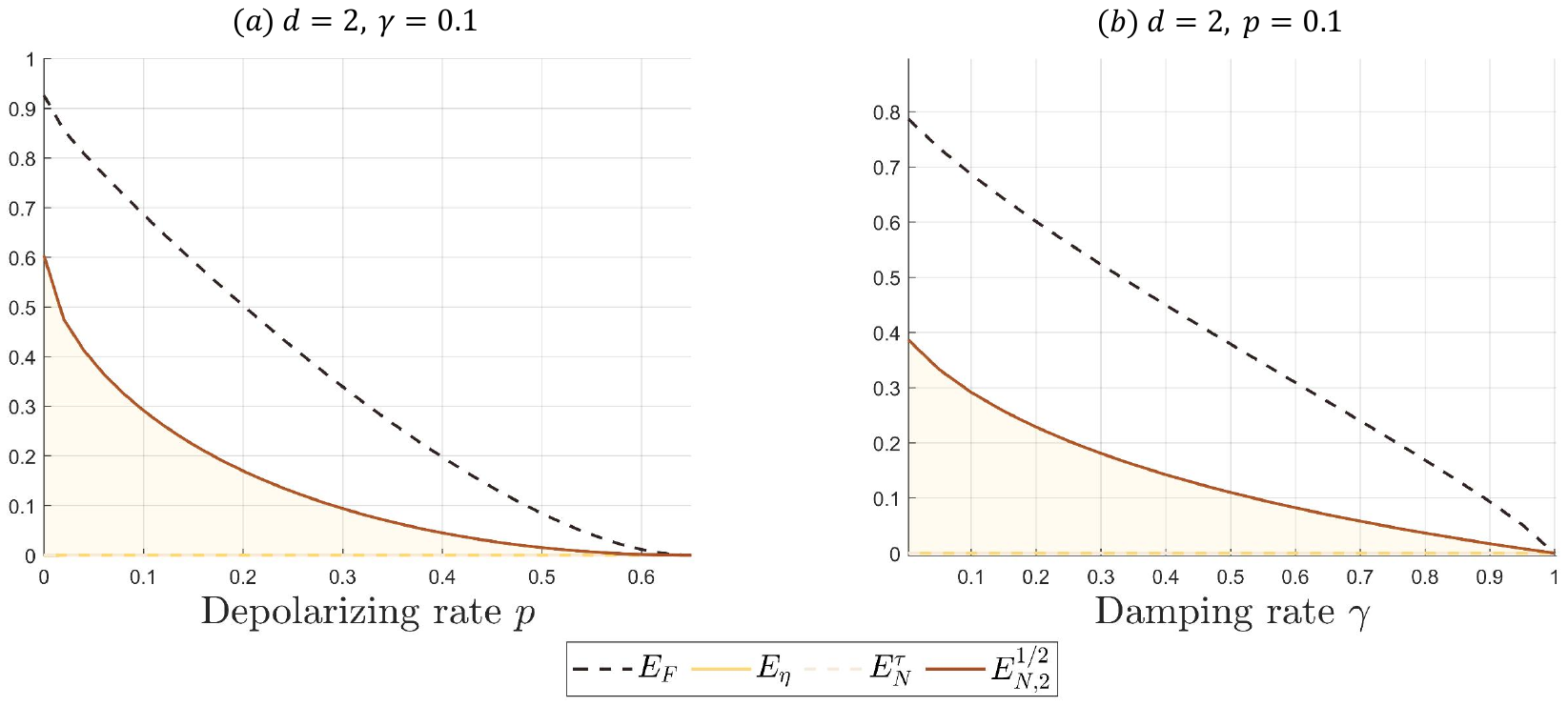}
    \caption{\textbf{Different bounds on the entanglement cost of noisy Bell states.} Panel (a-b) depict different bounds on the entanglement cost of $\rho_{AB} = \cA_{A\rightarrow A'}\ox \cD_{B\rightarrow B'}(\Phi_{AB}^{+})$. In Panel (a), we set $\gamma=0.1$. The $x$-axis represents the change of the depolarizing noise $p$. In Panel (b), we set $p=0.1$. The $x$-axis represents the change of the amplitude damping noise $\gamma$. The entanglement of formation $E_F(\cdot)$ provides an upper bound on $E_C(\cdot)$. $E_\eta(\cdot)$ and $E_{N}^\tau(\cdot)$ are computable lower bounds given in Ref.~\cite{Wang2016d} and Ref.~\cite{Lami2023a}, respectively. It shows that $E_{\NB,2}^{1/2}(\cdot)$ outperforms the other two lower bounds, both of which vanish for noisy Bell states having full rank in the case $\gamma >0$.}\label{fig:bound_comparison_qubit_qutrit_noisy_mes}
\end{figure}

\section{Comparison to other computable bounds}\label{appendix:comp_to_other_computable_bounds}
In fact, due to the faithfulness of the logarithmic fidelity of binegativity, i.e., $E_{\NB,2}^{1/2}(\rho_{AB})>0$ for any $\NPT$ state $\rho_{AB}$, it could be more practical and valuable compared with previous efficiently computable lower bounds $E_{\eta}(\cdot)$ and the \textit{tempered negativity bound} in Ref.~\cite{Lami2023a}.
We show the effectiveness of $E_{\NB,2}^{1/2}(\cdot)$ for lower bounding the entanglement cost of some general noisy states. Consider a two-qubit quantum state
\begin{equation}
    \rho_{AB} = \cA_{A\rightarrow A'}\ox \cD_{B\rightarrow B'}(\Phi_{AB}^{+}),
\end{equation}
where $\Phi_{AB}^+$ is the Bell state, $\cA_{A\rightarrow A'}$ is an amplitude damping channel with Kraus operators $K_0 = \ketbra{0}{0} + \sqrt{1-\gamma}\ketbra{1}{1}, K_1 = \sqrt{\gamma}\ketbra{0}{1}$, and $\cD_{B\rightarrow B'}$ is a qubit depolarizing channel such that $\cD(\rho) = p \rho + (1-p)I/2$.

In Fig.~\ref{fig:bound_comparison_qubit_qutrit_noisy_mes} (a), we can see that when $\gamma=0.1$ and the depolarizing rate $p$ varies from $0$ to $0.66$, $E_{\NB,2}^{1/2}(\rho_{AB})$ is strictly larger than previous bound $E_{\eta}(\rho_{AB})$ and $E_{N}^\tau(\rho_{AB})$, both of which are trivially vanishing. When $p$ is larger than 0.66, the state becomes a PPT state, thus having a zero entanglement cost. Our bound exhibits a smaller gap between the entanglement of formation, which serves as an upper bound for $E_C(\rho_{AB})$, implying a more accurate estimation of the entanglement cost. 
In Fig.~\ref{fig:bound_comparison_qubit_qutrit_noisy_mes} (b), we see that when $p=0.1$ and $\gamma$ varies from $0$ to $1$, our bound outperforms the previous bounds and serves as a better estimation on the entanglement cost of $\rho_{AB}$.

%%%%%%%%%%%%%%%%%%%%%%%%%%%%%%%%%%%%%%%%%%%%%%%%%%
%%%%%%%%%%%%%%%%%%%%%%%%%%%%%%%%%%%%%%%%%%%%%%%%%%
\section{Entanglement cost of quantum channels}\label{appendix:dynamical}

We provide a variational method to approximate $\widehat{E}_{\NB, 2}^{1/2}(\cN_{A\rightarrow B})$, thus can potentially serve as an improved bound of $E_{\NB,2}^{1/2}(J_{\cN})$ by cases.
To showcase the potential of variational assisted lower bound. We construct parameterized quantum circuits (PQCs) to generate general pure states
that live in the composite system $AA'$. Note that such PQCs are useful in developing near-term quantum algorithms~\cite{Benedetti2019a,Cerezo2020,Chen2020a,Higgott2018,Wang2020d,Bharti2021,Endo2020a} and quantum information processing protocols~\cite{Zhao2021,Khatri2018,Meyer2021}. 

Specifically, given a $n$-qubit channel $\cN_{A\rightarrow B}$, we utilize a $2n$-qubit PQC to parameterize a state $\ket{\psi_{AA'}(\bm \theta)}$. We use the following PQC structure illustrating an $n=2$ example for single-qubit channels with in total $D=10$ repeating blocks.

\begin{figure}[t]
    \centering
    \includegraphics[width=.5\linewidth]{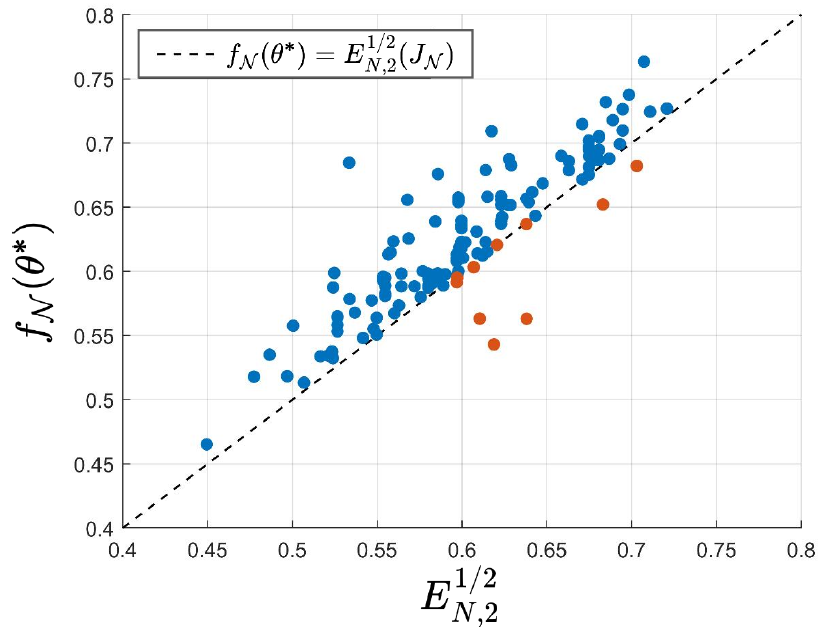}
    \caption{\textbf{Comparison between lower bounds for entanglement cost of two-qubit random mixed unitary channels.} Each point corresponds to a mixed unitary channel with four randomly generated unitaries according to the Haar measure. The $x$ coordinate of this point corresponds to $E_{N,2}^{1/2}(J_{\cN})$ and the $y$ coordinate of this point corresponds to $f_{\cN}({\bm \theta^*})$. The dashed line indicates the channels that hold $E_{N,2}^{1/2}(J_{\cN}) = f_{\cN}({\bm \theta^*)}$.}
    \label{fig:random_channel_bound_comparison}
\end{figure}

\begin{equation*}\label{Cir:hardware_efficient}
    \begin{array}{c}
        \centering
        \Qcircuit @C=.2em @R=.02em{
        & & & & & & & & & & & & & & \quad\quad\times D \\
        &\qw &\gate{R_y(\theta_{1,1})} &\ctrl{1} &\qw &\qw &\targ &\qw &\gate{R_y(\theta_{2,1})} &\ctrl{1} &\qw &\qw &\targ &\qw &\qw &\qw\\ 
        &\qw &\gate{R_y(\theta_{1,2})} &\targ &\ctrl{1} &\qw &\qw &\qw &\gate{R_y(\theta_{2,2})} &\targ &\ctrl{1} &\qw &\qw &\qw &\qw&\qw\\ 
        &\qw &\gate{R_y(\theta_{1,3})} &\qw &\targ &\ctrl{1} &\qw &\qw &\gate{R_y(\theta_{2,3})} &\qw &\targ &\ctrl{1} &\qw &\qw &\qw&\qw\\ 
        &\qw &\gate{R_y(\theta_{1,4})} &\qw &\qw &\targ &\ctrl{-3} &\qw &\gate{R_y(\theta_{2,4})} &\qw &\qw &\targ &\ctrl{-3} &\qw &\qw &\qw \gategroup{2}{9}{5}{13}{.5em}{--}
        \\ 
        }
    \end{array}
\end{equation*}
Here, $\bm \theta = (\theta_{1,1}, \theta_{1,2}, ...,\theta_{D+1,4})$, and the initial state of the circuit is set to $\ket{0}^{\ox n}$. The circuit acting on the initial state to produce the ansatz state $\ket{\psi_{AA'}(\bm \theta)}$. Then we set loss function $f_{\cN}(\bm{\theta})$ as
\begin{equation}
    f_{\cN}(\bm{\theta}) \coloneqq E_{\NB,2}^{1/2}\big((\cN_{A \rightarrow B} \otimes \cI_{A^{\prime}})\psi_{A A^{\prime}}(\bm{\theta})\big)
\end{equation}
and use gradient-free methods to maximize $f_{\cN}(\bm{\theta})$. As a consequence, we will get an optimized value $f(\bm{\theta}^*)$ after finite-iteration optimization and we have $E_C(\cN_{A \rightarrow B}) \ge f_{\cN}(\bm{\theta}^*)$. 

To show the potential improvement over $E_{\NB,2}^{1/2}(J_{\cN})$ offered by the proposed variational method, we conduct numerical experiments by random sampling two-qubit mixed unitary channels $\cU_{A\rightarrow B}(\rho) = \sum_{i=0}^3 p_i U_i \rho U_i^\dagger$, where $\sum_{i=0}^3 p_i= 1$ and $U_i$ are unitary operators on a two-qubit system. We set the parameters to be $[0.4, 0.4, 0.1, 0.1]$ and randomly generate unitaries according to the Haar measure. For each sampled mixed unitary channel, we compute the two lower bounds in Theorem 2. For the variational lower bound, we take $100$ optimization steps to reach the bound value at $\bm{\theta^*}$ for each sampled channel with a learning rate of $0.07$. In Fig.~\ref{fig:random_channel_bound_comparison}, we compare relative magnitudes between $f_{\cN}(\bm{\theta^*})$ and $E_{\NB,2}^{1/2}(J_\cN)$ where the black dashed line represents the cases $f_{\cN}(\bm{\theta^*})=E_{\NB,2}^{1/2}(J_\cN)$. 
Notably, for most of these mixed unitary channels, we have $f_{\cN}(\bm{\theta^*})>E_{\NB,2}^{1/2}(J_\cN)$.

% We now generalize our theory to the bipartite quantum channels. Similar to the point-to-point scenario, we first establish the following relationship between the entanglement cost of a bipartite quantum channel and its Choi state,  
% \begin{equation}\label{eq:bchannel_cost_lb_choi_state_method}
%     E_{C}(\cN_{AB\rightarrow A'B'}) \geq E_C(J_{AA'BB'}^{\cN}),
% \end{equation}
% where $J^{\cN}_{AA'BB'}$ is the Choi state of $\cN_{AB\rightarrow A'B'}$ and the bipartite cut is between $AA'$ and $BB'$ for $E_C(J_{AA'BB'}^{\cN})$.
\section{Lower bound on entanglement cost of bipartite channels}\label{appendix:bound_bi_channel}
We now generalize our theory to the bipartite quantum channels. 
\begin{definition}
Let $\varepsilon \geq 0$ and $\cN_{AB\rightarrow A'B'}$ be a bipartite channel. The one-shot entanglement cost of $\cN_{AB\rightarrow A'B'}$ with error $\varepsilon$ is defined as
\begin{equation}
    E_{C,\varepsilon}^{(1)}(\cN) \coloneqq \min\left\{\log k \, :\, \left\|\cN_{AB\rightarrow A'B'} - \cL_{A B \Bar{A}\Bar{B} \rightarrow A' B'}\big(\cdot \ox \Phi_{\Bar{A}\Bar{B}}^+(2^k)\big)\right\|_{\diamond}\leq \varepsilon, \, k\in \mathbb{N} \right\},
\end{equation}
where the minimization ranges over all LOCC operations $\cL_{A B \Bar{A}\Bar{B}}$ between Alice and Bob.
\end{definition}
The asymptotic entanglement cost of $\cN_{AB\rightarrow A'B'}$ is then defined as~\cite{Gour2020}, 
\begin{equation}
    E_C(\cN) \coloneqq \lim_{\varepsilon \rightarrow 0} \liminf_{n\rightarrow \infty} \frac{1}{n}E_{C, \varepsilon}^{(1)}(\cN^{\ox n}).
\end{equation}
We first establish the following relationship between the entanglement cost of a bipartite quantum channel and its Choi state. 
\begin{lemma}\label{lem:bchannel_cost_lb_choi_state}
    For a bipartite quantum channel $\cN_{AB\rightarrow A'B'}$,
    \begin{equation}
        E_{C}(\cN_{AB\rightarrow A'B'}) \geq E_C(J_{AA'BB'}^{\cN}),
    \end{equation}
    where $J^{\cN}_{AA'BB'}$ is the Choi state of $\cN_{AB\rightarrow A'B'}$ and the bipartite cut is between $AA'$ and $BB'$ for $E_C(J_{AA'BB'}^{\cN})$.
\end{lemma}
\begin{proof}
    Let us start with the one-shot scenario for simulating one use of $\cN_{AB\rightarrow A'B'}$. Suppose $\cL_{A B\Bar{A}\Bar{B}\rightarrow A' B'}$ is the optimal LOCC operation in Eq.~\eqref{Eq:channel_entcost_def} for $\cN_{AB\rightarrow A'B'}$ with a one-shot entanglement cost $k$. It follows that
    \begin{equation}
        \max_{\rho_{\hat{A}\hat{B}AB}}\bigg\| \cI_{\hat{A}\hat{B}}\ox \cL_{AB\Bar{A}\Bar{B} \rightarrow A' B'}\big(\rho_{\hat{A}\hat{B}AB} \ox \Phi_{\Bar{A}\Bar{B}}^{+}(2^k)\big) - \cI_{\hat{A}\hat{B}}\ox \cN_{AB\rightarrow A'B'}(\rho_{\hat{A}\hat{B}AB}) \bigg\|_1 \leq \epsilon.
    \end{equation}
    Then considering the Choi state of $\cN_{AB\rightarrow A'B'}$ with a form $J^{\cN}_{\hat{A}A'\hat{B}B'} = (\cI_{\hat{A}\hat{B}}\ox \cN_{AB\rightarrow A'B'}) (\Phi^+_{\hat{A}A}\ox\Phi^+_{\hat{B}B})$, we have a LOCC protocol $\Lambda$ from $A$ to $B$ such that
    \begin{equation}
    \begin{aligned}
        & \bigg\| \Lambda\big(\Phi_{\Bar{A}\Bar{B}}^+(2^k)\big) - J^{\cN}_{\hat{A}A'\hat{B}B'}\bigg\|_1 \\
        = & \bigg\| (\cI_{\hat{A}\hat{B}}\ox \cL_{\Bar{A}\Bar{B} A B \rightarrow A' B'})\big(\Phi^+_{\hat{A}\hat{B}AB} \ox \Phi^{+}_{\Bar{A}\Bar{B}}(2^k)\big) - (\cI_{\hat{A}\hat{B}}\ox \cN_{AB\rightarrow A'B'})(\Phi^+_{\hat{A}\hat{B}AB}) \bigg\|_1 \\
        \leq &\; \epsilon,
    \end{aligned}
    \end{equation}
    where we abbreviate $\Phi^+_{\hat{A}\hat{B}AB} = \Phi^+_{\hat{A}A}\ox\Phi^+_{\hat{B}B}$. Hence, by the definition of the one-shot entanglement cost of quantum states, we have $\Lambda$ is a feasible protocol for $E_{C,\epsilon}(J^{\cN}_{\hat{A}A'\hat{B}B'})$ with consumption of $k$ ebits. Considering all possible protocols for converting the maximally entangled state to $J^{\cN}_{\hat{A}A'\hat{B}B'}$, we have $E_{C,\epsilon}^{(1)}(\cN) \geq E_{C,\epsilon}^{(1)}(J^{\cN}_{\hat{A}A'\hat{B}B'})$. Similarly, we can derive that for any $n \in \ZZ_+$ and $\varepsilon \in (0,1)$, 
    \begin{equation}
        E_{C,\epsilon}^{(1)}(\cN^{\ox n}) \geq E_{C,\epsilon}^{(1)}\big(J^{\cN}_{\hat{A}A'\hat{B}B'}\big),
    \end{equation}
    As a consequence, taking $n$ to arbitrarily large, the $\varepsilon$ to arbitrarily small, we have
    \begin{equation}\label{eq:one_shot_bipartite_cost_lb}
        E_{C}(\cN) = \lim_{\varepsilon \rightarrow 0} \liminf_{n\rightarrow \infty} \frac{1}{n} E_{C, \varepsilon}^{(1)}(\cN^{\ox n}) \geq \lim_{\varepsilon \rightarrow 0} \liminf_{n\rightarrow \infty} \frac{1}{n} E_{C, \varepsilon}^{(1)}\left((J^{\cN}_{\hat{A}A'\hat{B}B'})^{\ox n}\right)= E_C(J^{\cN}_{\hat{A}A'\hat{B}B'}),
    \end{equation}
    as desired.
\end{proof}

\begin{figure}[t]
    \centering
    \includegraphics[width = 0.45\linewidth]{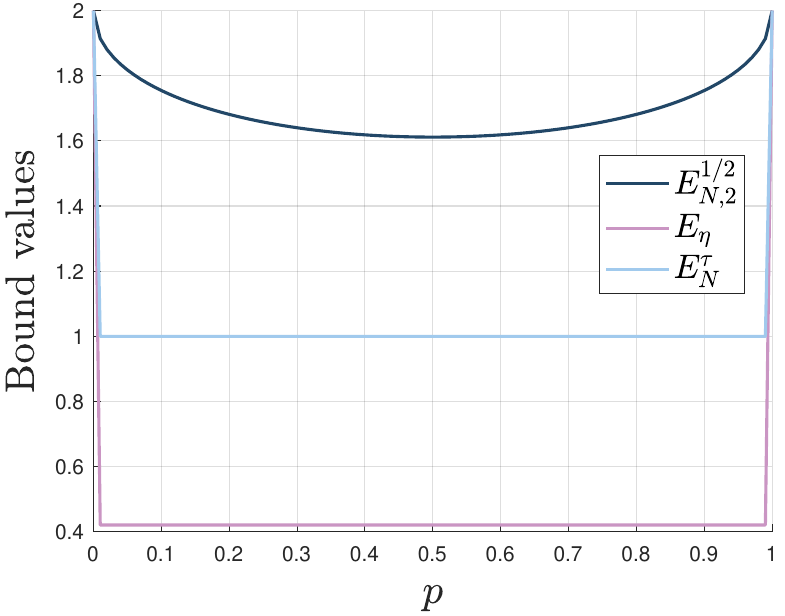}
    \includegraphics[width = 0.54\linewidth]{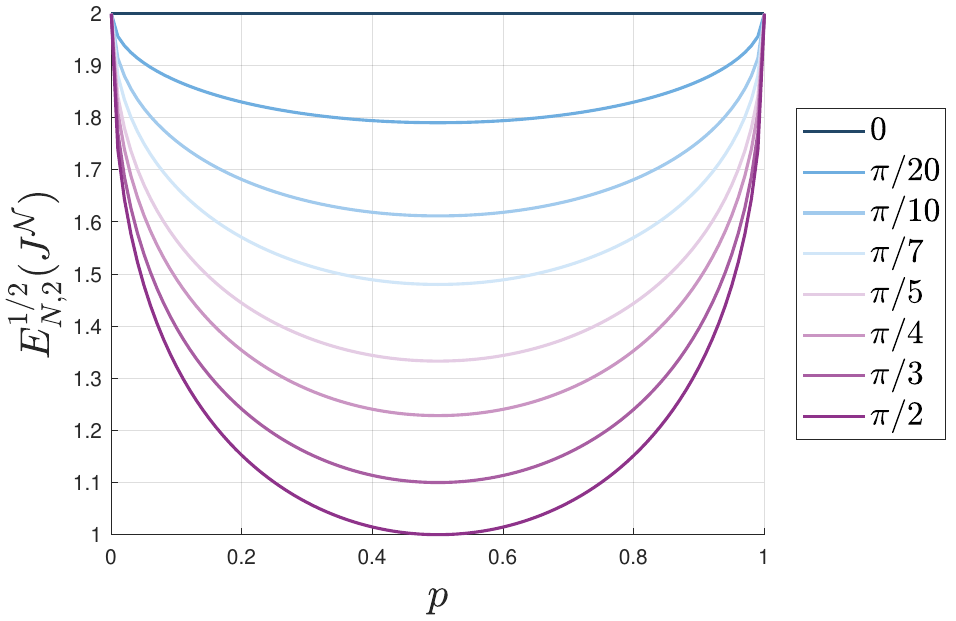}
    \caption{\textbf{Illustration on Choi state lower bounds for entanglement cost of the collective dephased SWAP operation $\cN_p$ versus the channel parameter $p$.} (Left) Comparison of the entanglement cost lower bounds $E_{\eta}(J^{\cN_p})$~\cite{Wang2016d}, $E^{\tau}_{N}(J^{\cN_p})$~\cite{Lami2023a} and $E^{1/2}_{N,2}(J^{\cN_p})$ based on the Choi state lemma, where the collective phase $\phi = \pi/10$; (Right) The variations of the logarithmic fidelity of binegativity lower bound with respect to different collective phases $\phi$.}
    \label{fig:noisy_bipartite_cost}
\end{figure}

\begin{remark}[~\cite{Das2020}]\label{rem:bicovariant_channel_cost}
    For any bipartite channel $\cN_{AB\rightarrow A'B'}$ that is bicovariant, then the (sequential) entanglement cost is given by
    $$E_{C}(\cN) = E_{C}(J^{\cN}_{AA'BB'})$$
\end{remark}

In the following, we also perform numerical calculations to illustrate the behaviors of our lower bound on the entanglement cost of bipartite channels relying on the logarithmic fidelity of binegativity of their Choi states. Exceptionally, we consider a qubit SWAP operation undertaken a collective dephasing noise by a phase $\phi$, which was first studied in the dissipation of quantum computation~\cite{Palma1996quantum}. The composite channel $\cN_p$ is defined as,
\begin{equation}
    \cN_p(\rho_{AB}) \coloneqq p \rho^S_{AB} + (1-p)U_{\phi} \rho^S_{AB}U^{\dagger}_{\phi}, \; U_{\phi} = 
    \begin{pmatrix}
        1 & 0 & 0 & 0\\
        0 & e^{i\phi} & 0 & 0\\
        0 & 0 & e^{i\phi} & 0\\
        0 & 0 & 0 & e^{2i\phi}
    \end{pmatrix},
\end{equation}
where $\rho^{S}_{AB} = \Op{SWAP} \rho_{AB} \Op{SWAP}$ for any $\rho_{AB}\in \mathscr{D}(\cH_{A}\ox\cH_{B})$. In terms of the Choi state lower bound, we have applied the state entanglement cost lower bounds discussed in previous sections~\cite{Lami2023a,Wang2016d} to the Choi state of bipartite channels regarding the partition $AA':BB'$, which also form entanglement cost lower bounds for bipartite channels by Lemma~\ref{lem:bchannel_cost_lb_choi_state}.

Our results are shown in Fig.~\ref{fig:noisy_bipartite_cost}. Compared to the other two bounds at $\phi = \pi/10$,  $E_{\NB,2}^{1/2}$ outperforms $E_{\eta}$ and $E_{N}^{\tau}$ ranging over almost all $p$ values in $(0,1)$, where the tempered negativity bound keeps a value around one and $E_{\eta}$ becomes even less than one. On the other hand, we observe a similar trend recorded from Ref.~\cite{Bauml2018} for the collective dephasing SWAP operation where the cost reaches the lowest at $p = 1/2$ due to the most uncertainty about whether the dephasing occurs. Interestingly, when $\phi = \pi/2$\footnote{Equivalent to a collective phase rotation by $\pi$ in Ref.~\cite{Bauml2018}}, our $E^{1/2}_{\NB,2}$ also achieves a reduction of a factor of $1/2$ at $p=1/2$, compared to the cost for an ideal SWAP operation.

\section{Alternative lower bound of Entanglement cost}\label{sec:alternative_lb}
We have also developed an alternative lower bound of entanglement cost using a similar methodology as for the logarithmic fidelity of binegativity bound.  
In particular, we can define another sub-state set which we denote $\PPT^\#_{k}(A:B)$ in the expression,
\begin{equation}
    \PPT^\#_{k}(A:B) := \Big\{\omega_1\in\mathscr{P}(\cH_{AB}): \exists\{\omega_i\}_{i=2}^k,\ \text{s.t.}\ \big|\omega_i^{T_B}\big| \preceq \omega_{i+1}, \forall i\in[1:k-1], \big\|\omega_{k}^{T_B}\big\|_1\le 1\Big\}.
\end{equation}
The set resembles $\Freek$ by replacing the $|\cdot|_*$ operation with the usual matrix absolute operation where $|A| = \sqrt{A^{\dagger} A}$ for any matrix $A$. One can also prove a similar inclusion relationship within the hash hierarchy. Note that this is the original $\PPT_k$ in the v1 and v2 of this arxiv manuscript. 
We sincerely thank the TQC2025 Program Committee for highlighting that our SDP representation of this set in the previous version was not tight. In this revision, we clarify the useful properties of $\PPT^{\#}_k(A:B)$ and demonstrate how these properties can be effectively leveraged to establish a semidefinite programming (SDP) lower bound on the entanglement cost.

\begin{proposition}
For a bipartite system $\cH_A\ox \cH_B$ and a positive integer $k \geq 2$, it holds that  
\begin{equation}
    \PPT(A:B) \subsetneq \PPT^{\#}_k(A:B) \subsetneq \cdots \subsetneq \PPT'(A:B).
\end{equation}
\end{proposition}
\begin{proof}
For any fixed $j$ and any sub-state $\sigma_{AB}\in \PPT^{\#}_{j}$, we have $|\sigma_{AB}^{T_B}| \leq \omega_2, |\omega_2^{T_B}| \leq \omega_3, \cdots,|\omega_{j-1}^{T_B}| \leq \omega_{j},\|\omega_{j}^{T_B}\|_1\le 1$. It follows that $\|\omega_{j-1}^{T_B}\|_1 \leq \tr\omega_j \leq \|\omega_{j}^{T_B}\|_1 \leq 1$ which indicates that $\sigma_{AB}\in \PPT^{\#}_{j-1}$.
\end{proof}
\begin{lemma}\label{lem:LB_leq0_PPT_alter}
For any non-negative integer $k$, a bipartite quantum state $\rho_{AB}$ is in $\PPT^{\#}_k$ if and only if $\rho_{AB}\in \PPT$.
\end{lemma}
\begin{proof}
For the `if' part, if $\rho_{AB}\in \PPT$, we have $\rho_{AB}^{T_B}=\omega_1\geq 0$. Then we can construct a sequence of $\omega_{j}$ as $\omega_k = \rho_{AB}$ if $k$ is odd; $\omega_{k} = \omega_1$ if $k$ is even. Therefore, it follows that $\|\omega_{k}^{T_B}\|_1 \leq 1$ and $\rho_{AB}\in\PPT^{\#}_k$. For the `only if' part, supposing there is a $\rho_{AB}\in\PPT^{\#}_k$ and $\rho_{AB}\notin \PPT$, it follows that $\tr \omega_{2}\geq \tr|\rho_{AB}^{T_B}| > 1$. Notice that $\|\omega_k^{T_B}\|_1 \geq \tr\omega_2 > 1$, a contradiction. Hence the proof.
\end{proof}

{
% To clarify, the set $\PPT^{\#}_k(A:B)$ is firstly defined in our previous version of the manuscript, originally denoted as `$\PPT_k(A:B)$'. However, based on the hash hierarchy, one has to make a non-trivial relaxation to construct a superadditive lower bound of entanglement cost under NPT theory. We start with the relationship between $\PPT^{\#}_k(A:B)$ and $\PPT_k(A:B)$.
\begin{lemma}
    For a positive integer $k\geq 2$, it holds that $\PPT^{\#}_k(A:B) \subseteq \PPT_k(A:B)$, where 
    \begin{align}
    \Freek \coloneqq \Big\{\omega_1 \in\mathscr{L}(\cH_{AB}): \omega_1\succeq 0,~\exists \{\omega_i\}_{i=2}^{k},~\mathrm{s.t.}
    |\omega_i^{T_B}|_* \preceq \omega_{i+1}, \forall i\in [1:k-1],~\big\|\omega_{k}^{T_B}\big\|_1\le 1\Big\}
    \end{align}
    and $|X|_{*} \preceq Y$ denotes $-Y \preceq X\preceq Y$.
\end{lemma}
\begin{proof}
    Let $\omega_1$ be an element of $\PPT^{\#}_k(A:B)$ such that $\exists \{\omega_i\}_{i=2}^k$ satisfying $|\omega_i^{T_B}| \preceq \omega_{i+1}$. This then gives, $|\omega_i^{T_B}| - \omega_{i+1} \preceq 0 \preceq \omega_{i+1} - |\omega_i^{T_B}|$. Notice that $-|\omega_i^{T_B}| \preceq \omega_i^{T_B} \preceq |\omega_i^{T_B}|$. Add this inequality to the previous one to derive $-\omega_{i+1} \preceq \omega_i^{T_B} \preceq \omega_{i+1}$, and hence, $\omega_1\in \PPT_k(A:B)$ as required.
\end{proof}
}

{We can also define an entanglement quantifier called the \textit{generalized divergence of hash $k$-negativity} of a given bipartite quantum state $\rho_{AB}$ acting on the composite Hilbert space $\cH_{AB}$ as,
\begin{equation}
    E_{\#, k}(\rho_{AB}) = \min_{\sigma_{AB}\in \PPT^\#_{k}(A:B)} \mathbf{D}(\rho_{AB}||\sigma_{AB}).
\end{equation}
The quantifier can be shown to be faithful regarding the $\PPT$ states since it is never negative and equals zero if and only if $\rho_{AB}$ is $\PPT$. 
\begin{lemma}[Faithfulness of $E_{\#,k}$]\label{lem:faithfulness_alter}
    For a bipartite quantum state $\rho_{AB}\in\mathscr{D}(\cH_A\ox\cH_B)$, $E_{\#,k}(\rho_{AB})\geq 0$ and $E_{\#,k}(\rho_{AB}) = 0$ if and only if $\rho_{AB}\in \PPT(A:B)$.
\end{lemma}
\begin{proof}
    $E_{\#,k}(\rho_{AB})\geq 0$ can be directly obtained by $\mathbf{D}(\rho\|\sigma) \geq 0$ when $\tr \sigma \leq 1$. For the ‘if’ part, since $\PPT\subsetneq\PPT^{\#}_k$, we have $E_{\#,k}(\rho_{AB}) \le \mathbf{D}(\rho_{AB}||\rho_{AB}) = 0$ which gives $E_{\#,k}(\rho_{AB}) = 0$. For the ‘only if’ part, since $\mathbf{D}(\rho_{AB}\|\sigma_{AB}) = 0$ if and only if $\rho_{AB} = \sigma_{AB}$, we conclude $\rho_{AB}\in \PPT$ by Lemma~\ref{lem:LB_leq0_PPT_alter}.
\end{proof}}

{
With the generalized divergence is the \textit{sandwiched R{\'{e}}nyi relative entropy}~\cite{Muller_Lennert2013,Wilde2014a}. We can obtain a family of \textit{R{\'{e}}nyi-$\alpha$ divergence of hash $k$-negativity} as 
\begin{equation}
    E_{\#, k}^{\alpha}(\rho_{AB}) = \min_{\sigma_{AB}\in\PPT^\#_{k}(A:B)} \widetilde{D}_{\alpha}(\rho_{AB}||\sigma_{AB}).
\end{equation}
In particular, taking $k=2$, $\alpha=1/2$, we define the \textit{hash logarithmic fidelity of binegativity} as 
\begin{equation}\label{eq:hash_binegativity}
   E_{\#,2}^{1/2}(\rho_{AB}) = - \log \max_{\sigma_{AB}\in\PPT^{\#}_2(A:B) }F(\rho_{AB},\sigma_{AB}),
\end{equation}
where $F(\rho, \sigma)=\left(\tr\left[(\sqrt{\sigma}\rho\sqrt{\sigma})^\frac{1}{2}\right]\right)^2$ is the Uhlmann’s fidelity between $\rho$ and $\sigma$. In contrast with the no-hash scenario, the quantifier $E^{1/2}_{\#,2}$ cannot be directly estimated via SDP. Despite that, we can modify the optimization problem in Eq.~\eqref{eq:hash_binegativity} and make a proper lower bound of $E^{1/2}_{\#,2}$ that can be estimated via SDP.}

\begin{proposition}
For a bipartite quantum state $\rho_{AB}\in\mathscr{D}(\cH_A\ox\cH_B)$, its entanglement cost is lower bounded by
\begin{subequations}
\begin{align*}
    E_{C}(\rho_{AB})\geq \ECPPT(\rho_{AB})\geq E^{1/2}_{\#,2}(\rho_{AB}) \geq -2\log \widehat{f}_{\NB,2}(\rho_{AB}) \coloneqq -2\log \max & \;\; \frac{1}{2}\tr(X_{AB} + X_{AB}^\dagger)\\
     {\rm s.t.}   &\;\; C_{AB},D_{AB},M_{AB},N_{AB}\succeq 0,\\
     &\;\; C_{AB} + D_{AB} = M_{AB}^{T_B} - N_{AB}^{T_B},\\
     &\; \tr(M_{AB} + N_{AB}) \leq 1,\\
     &\; \left(\begin{array}{cc}
         \rho_{AB} & X_{AB} \\
         X_{AB}^\dagger & C^{T_B}_{AB} - D_{AB}^{T_B}
     \end{array}\right) \succeq 0.
\end{align*}
\end{subequations}
\end{proposition}

\begin{proof}
Followed from the same idea in Theorem~\ref{thm:min_thm}, we have
\begin{equation*}
    E_{R,\PPT}(\rho_{AB}) \geq - \max_{\sigma_{AB}\in\PPT(A:B)} \log F(\rho_{AB},\sigma_{AB}),
\end{equation*}
where $E_{R,\PPT}(\rho_{AB})$ is the PPT relative entropy of entanglement. 
% We have the following SDP representation
% \begin{subequations}
% \begin{align*}
%     -\max_{\sigma_{AB}\in\PPT} \log F(\rho_{AB},\sigma_{AB}) = -2\log \max & \;\; \frac{1}{2}\tr(X_{AB} + X_{AB}^\dagger)\\
%      {\rm s.t.}   &\;\; \sigma_{AB}\succeq 0,~\tr\sigma_{AB}=1,~\sigma_{AB}^{T_B} \succeq 0,\\
%      &\; \left(\begin{array}{cc}
%          \rho_{AB} & X_{AB} \\
%          X_{AB}^\dagger & \sigma_{AB}
%      \end{array}\right) \succeq 0.
% \end{align*}
% \end{subequations}
Based on  Lemma~\ref{lem:faithfulness_alter}, we can have that,
\begin{subequations}
\begin{align*}
    -\max_{\sigma_{AB}\in\PPT(A:B)} \log F(\rho_{AB},\sigma_{AB}) \geq -\max_{\sigma_{AB}\in\PPT^{\#}_2(A:B)} \log F(\rho_{AB},\sigma_{AB}) = -2\log \max & \;\; \frac{1}{2}\tr(X_{AB} + X_{AB}^\dagger)\\
     {\rm s.t.}   &\;\; \sigma_{AB}\succeq 0,~|\sigma_{AB}^{T_B}|\preceq \omega_{AB},\\ &\;\; ~\tr|\omega_{AB}^{T_B}| \leq 1,\\
     &\; \left(\begin{array}{cc}
         \rho_{AB} & X_{AB} \\
         X_{AB}^\dagger & \sigma_{AB}
     \end{array}\right) \succeq 0.
\end{align*}
\end{subequations}

For any feasible solution $\{X_{AB},\sigma_{AB}, \omega_{AB}\}$, {let $\sigma^{T_B}_{AB} = P_{AB} - Q_{AB}, \omega_{AB}^{T_B} = R_{AB} - S_{AB}$ where $P_{AB},Q_{AB}, R_{AB}, S_{AB} \succeq 0$ are the eigen-decompositions of $\sigma^{T_B}_{AB}$ and $\omega^{T_B}_{AB}$, respectively. We observe that $\{C_{AB} = P_{AB}, D_{AB} = Q_{AB}, M_{AB} = R_{AB},N_{AB}=S_{AB}\}$} is a feasible solution of the following SDP.
\begin{subequations}
\begin{align*}
\widehat{f}_{\NB,2}(\rho_{AB}) = \max & \;\; \frac{1}{2}\tr(X_{AB} + X_{AB}^\dagger)\\
{\rm s.t.}   &\;\; C_{AB},D_{AB},M_{AB},N_{AB}\succeq 0,\\
&\;\; C_{AB} + D_{AB} = M_{AB}^{T_B} - N_{AB}^{T_B},\\
&\; \tr(M_{AB} + N_{AB}) \leq 1,\\
&\; \left(\begin{array}{cc}
\rho_{AB} & X_{AB} \\
X_{AB}^\dagger & C^{T_B}_{AB} - D_{AB}^{T_B}
\end{array}\right) \succeq 0.
\end{align*}
\end{subequations}
{Therefore, we have $E^{1/2}_{\#,2}(\rho_{AB}) \geq -2\log\widehat{f}_{\NB,2}(\rho_{AB})$.} Now, we will show $\widehat{f}_{\NB,2}(\rho_{AB})$ is submultiplicative under tensor product, i.e.,
\begin{equation}
    \widehat{f}_{\NB,2}(\rho_0\ox\rho_1) \leq \widehat{f}_{\NB,2}(\rho_0)  \widehat{f}_{\NB,2}(\rho_1),
\end{equation}
by first writing out its SDP dual programming. By the Lagrangian method, we have the following dual problem.
\begin{equation}\label{SDP:append_dual}
\begin{aligned}
    \min & \;\; t\\
    {\rm s.t.} &\;\; \tr[Q_{AB}\rho_{AB}] = t,\\
    &\;\; -S_{AB} \preceq R_{AB}^{T_B} \preceq S_{AB},\\
    &\;\; -t I_{AB} \preceq S_{AB}^{T_B} \preceq t I_{AB},\\
    &\;\; \left(\begin{array}{cc}
         Q_{AB} & -I_{AB} \\
         -I_{AB} & R_{AB}
     \end{array}\right) \succeq 0,
\end{aligned}
\end{equation}
By Slater's condition, we note that the above two optimization programs satisfy strong duality, and both evaluate to $\widehat{f}_{\NB,2}(\rho_{AB})$.

If we assume the optimal solution to the SDP~\eqref{SDP:append_dual} for $\widehat{f}_{\NB,2}(\rho_0)$ and $\widehat{f}_{\NB,2}(\rho_1)$ are $\{Q_0,R_0,S_0\}$ and $\{Q_1,R_1,S_1\}$, respectively, then it follows
\begin{equation*}
    (Q_0\ox Q_1)(R_0\ox R_1) \succeq I\ox I,
\end{equation*}
where we used the fact that $Q_0R_0 \succeq I, Q_1R_1 \succeq I\implies Q_0R_0\ox Q_1R_1 \succeq I\ox Q_1R_1 \succeq I\ox I$. Next, it is easy to check that
\begin{equation*}
\begin{aligned}
& S_0 \ox S_1 + R_0^{T_B} \ox R_1^{T_{B'}} = \frac{1}{2}[(S_0 + R_0^{T_B})\ox (S_1 + R_1^{T_B'}) + (S_0 - R_0^{T_B})\ox (S_1 - R_1^{T_{B'}})] \succeq 0\\
& S_0 \ox S_1 - R_0^{T_B} \ox R_1^{T_{B'}} = \frac{1}{2}[(S_0 + R_0^{T_B})\ox (S_1 - R_1^{T_B'}) + (S_0 - R_0^{T_B})\ox (S_1 + R_1^{T_{B'}})] \succeq 0.
\end{aligned}
\end{equation*}
Then we have $-S_0 \ox S_1\preceq R_0^{T_B} \ox R_1^{T_{B'}} \preceq S_0 \ox S_1$. Also, we note that
\begin{equation*}
\begin{aligned}
    \|S_0^{T_B} \ox S_1^{T_B}\|_{\infty} \leq \|S_0^{T_B}\|_{\infty}\|S_1^{T_B}\|_{\infty}\leq \tr(Q_0\rho_0)\tr(Q_1\rho_1),
\end{aligned}
\end{equation*}
which gives 
\begin{equation*}
\begin{aligned}
& -\tr[(Q_0\ox Q_1)(\rho_0\ox \rho_1)]\cdot I\ox I \preceq S_0^{T_B} \ox S_1^{T_B} \preceq \tr[(Q_0\ox Q_1)(\rho_0\ox \rho_1)]\cdot I\ox I.
\end{aligned}
\end{equation*}
Therefore, $\{Q_0\ox Q_1, R_0\ox R_1, S_0\ox S_1\}$ is a feasible solution to the SDP of $\widehat{f}_{\NB,2}(\rho_0\ox\rho_1)$ which implies
\begin{equation}\label{Eq:F_sup_add_alt}
    \widehat{f}_{\NB,2}(\rho_0\ox\rho_1) \leq \tr(Q_0\rho_0)\tr(Q_1\rho_1) = \widehat{f}_{\NB,2}(\rho_0)  \widehat{f}_{\NB,2}(\rho_1).
\end{equation}
Consequently, we have
\begin{equation}
\begin{aligned}
    E_{C,\PPT}(\rho_{AB}) &\geq \lim_{n\rightarrow \infty} \frac{1}{n}E_{R}(\rho_{AB}^{\ox n})\\
    &\geq \lim_{n\rightarrow \infty} \frac{1}{n} -\log\widehat{f}_{\NB,2}(\rho_{AB}^{\ox n}) \geq -\log\widehat{f}_{\NB,2}(\rho_{AB}),
\end{aligned}
\end{equation}
where the first inequality follows from Ref.~\cite[p. 421]{Hayashi2006a} and the equality is a consequence of the submultiplicativity of $\widehat{f}_{\NB,2}(\cdot)$.
\end{proof}

\end{document}